\documentclass[runningheads]{llncs}

\usepackage{graphicx,amssymb,amsmath,mathtools}
\usepackage{graphicx}

\usepackage{todonotes}

\usepackage{xspace}

\usepackage{algorithm}
\usepackage{algorithmic}

\usepackage{tikz}
\usetikzlibrary{shapes,snakes}

\usepackage{url}

\newtheorem{fact}{Fact}

\newtheorem{assumption}{Assumption}

\newcommand\eat[1]{}
\newcommand\commentout[1]{}

\renewcommand\a{\alpha}
\newcommand\be{\beta}
\newcommand\ga{\gamma}
\newcommand\E{\mathcal{E}}

\newcommand\opt{\mathit{opt}}

\newcommand\cge{C_G}

\newcommand\B{\mathcal{B}} 
\newcommand\pp{\mathcal{P}} 
\newcommand\pr{\mathcal{R}} 
\newcommand\po{\mathcal{O}}

\newcommand\reals{\mathbb{R}}

\renewcommand\S{\mathcal{S}}
\newcommand\es{\mathcal{S}} 

\newcommand\crate{cut-rate\xspace}

\newcommand\dist{edge-distribution\xspace}
\newcommand\dists{edge-distributions\xspace}
\newcommand\mdist{maxmin-edge-distribution\xspace}
\newcommand\mdists{maxmin-edge-distributions\xspace}
\newcommand\pdist{prime-edge-distribution\xspace}
\newcommand\pdists{prime-edge-distributions\xspace}

\newcommand\edist{extreme-edge-distribution\xspace}

\newcommand\maxminpoly{maxmin-polytope\xspace}

\newcommand\CR{cr}

\newcommand\ocsg{omni-connected-spanning-subgraph\xspace}
\newcommand\ocsgs{omni-connected-spanning-subgraphs\xspace}
\newcommand\csg{connected spanning subgraph\xspace}
\newcommand\csgs{connected spanning subgraphs\xspace}
\newcommand\mcsg{minimum connected spanning subgraph\xspace}
\newcommand\mcsgs{minimum connected spanning subgraphs\xspace}

\urldef{\mailsa}\path|{h.aziz,oded, msp,rahul}@dcs.warwick.ac.uk|

\begin{document}

\mainmatter  

\title{Wiretapping a hidden network
\thanks{This research was supported in part by EPSRC projects
  EP/D067170/1, EP/G064679/1, and by the Centre for Discrete Mathematics and its Applications (DIMAP), EPSRC award EP/D063191/1.}}
\author{%
 Haris Aziz\inst{2} \and
 Oded Lachish\inst{1} \and
 Mike Paterson\inst{1}\and
 Rahul Savani\inst{3}}
\institute{%
 Department of Computer Science,
 University of Warwick,
 CV4 7AL Coventry, UK\\
 \email{\{oded,msp\}@dcs.warwick.ac.uk}
 \and 
 Institut f{\"u}r Informatik,
 Universit{\"a}t M{\"u}nchen,
 80538 M{\"u}nchen, Germany \\
 \email{aziz@tcs.ifi.lmu.de}
 \and
 Department of Computer Science,
 University of Liverpool,
 L69 3BX Liverpool, UK \\
 \email{rahul.savani@liverpool.ac.uk}}


%
%
%


%
%

\parskip0mm
\linespread{0.95}
\addtolength{\abovedisplayskip}{-1.2mm}
\addtolength{\belowdisplayskip}{-1.2mm}

\toctitle{Lecture Notes in Computer Science}
\tocauthor{Authors' Instructions}
\maketitle

\begin{abstract}
We consider the problem of maximizing
the probability of hitting a strategically chosen 
hidden {\em virtual network} by placing a wiretap on a single 
link of a communication network.
This can be seen as a two-player win-lose (zero-sum) game 
that we call the {\em wiretap game}.
The {\em value} of this game is the greatest probability 
that the wiretapper can secure for hitting the virtual network.
The value 
is shown to equal the reciprocal of the 
{\em strength} of the underlying graph. 

We efficiently compute a unique partition of the edges of the 
graph, called the prime-partition, and find the set of pure 
strategies of the hider that
are best responses against every maxmin strategy of the
wiretapper.
Using these special pure strategies of the hider, which
we call \ocsgs, we define a partial order on the elements of 
the prime-partition.
From the partial order, we obtain a linear number 
of simple two-variable inequalities that define the
\maxminpoly, 
and a characterization of its extreme points. 

Our definition of the partial order allows us to find
all equilibrium strategies of the wiretapper that minimize 
the number of pure best responses of the hider.  
Among these strategies, 
we efficiently compute the {\em unique} strategy that 
maximizes the least punishment that the hider incurs for 
playing a pure strategy that is not a best response.
%
Finally, we show that this unique strategy is the nucleolus of 
the recently studied simple cooperative {\em spanning
connectivity game}.
\end{abstract}


\section{Introduction}
\label{s-intro}

\eat{
\begin{enumerate}
\item
Types of questions we can answer on maxmin strategies
\item
References to similar games: One sentence Alon, Steve
Alpern, hiding on nodes. 
\end{enumerate}
}

Communication networks consist of two major layers,
large 
static
physical networks, and virtual networks 
built on top of them. 
The physical infrastructure comprises optical fibres,
circuits, and routers etc., and rarely changes.
A virtual network specifies how to route traffic between
nodes, is software-driven and hence flexible.
Modern physical networks are highly-connected and offer many
possibilities for constructing virtual networks.
Security is an important consideration for choosing a virtual
network.
One aspect of network security is resilience to wiretapping,
which is the problem we study here from a game-theoretic
perspective.

We consider the problem of maximizing
the probability of hitting a strategically chosen 
hidden {\em virtual network} by placing a wiretap on a single 
link of a communication network, represented by an 
undirected, unweighted graph. 
%
This can be seen as a two-player win-lose (zero-sum) game 
that we call the {\em wiretap game}.
A pure strategy of the wiretapper is an edge to tap, 
and of his opponent, the {\em hider}, a choice of virtual network, 
a {\em connected spanning subgraph}.
The {\em wiretapper} 
wins, with payoff one, when he picks an edge in the
network chosen by the hider, and loses, 
with payoff zero, otherwise. 
Thus, the {\em value} of this game is the greatest probability 
that the wiretapper can secure for hitting the hidden network.
He does this by playing a maxmin strategy, which is a 
probability distribution on the edges.
The value also equals the smallest probability that the hider
can secure, which she does by playing a minmax strategy, which is a
probability distribution on connected spanning subgraphs.

\paragraph{\bf Our results.}
The value of the wiretap game is shown to equal the reciprocal of the 
{\em strength} of the underlying graph, a concept introduced by
Gusfield~\cite{Gus83}. 
We efficiently compute a unique partition of the edges of the 
graph, called the \emph{prime-partition}. 
We find the set of pure strategies of the hider that
are best responses against every maxmin strategy of the
wiretapper.
Using these special pure strategies of the hider, which
we call {\em \ocsgs}, we define a partial order on the elements of 
the prime-partition.
Our definition
in terms of \ocsgs 
is central to proving our results.

From the partial order, we obtain a linear number 
of simple two-variable inequalities that define the
\maxminpoly, and a characterization of its extreme points.
In contrast, the natural description of the \maxminpoly 
is as the solutions to a linear program with exponentially many
constraints. 
%
%
Our definition of the partial order allows us to find
all equilibrium strategies of the wiretapper that minimize 
the number of pure best responses of the hider.  
Among these strategies, we efficiently compute the
{\em unique} strategy that maximizes the least 
punishment that the hider incurs for playing a 
pure strategy that is not a best response.
%

Finally, we show that our analysis of the wiretap game 
provides a polynomial-time algorithm 
for computing the nucleolus of the {\em spanning
connectivity game}, a simple cooperative game~\cite{algocoopth}.
In this game, the players are the edges of the graph and a
coalition, which is a subset of edges, has value one if it is
a connected spanning subgraph, and zero otherwise.
The characterization of the maxmin strategies of the wiretap
game carries over to the least-core polytope of the spanning
connectivity game, and the nucleolus of this game is the
special maxmin strategy we compute for the wiretap game.

\paragraph{\bf Related work.}

Wiretapping, as an important aspect of network security, has 
received recent attention in different settings, see 
e.g.~\cite{GY00} and~\cite{Jai04}.

The strength of an unweighted graph, which has a central
role in our work, is also called the edge-toughness,
and relates to the classical work of 
Nash-Williams~\cite{NashW} and Tutte~\cite{Tutte}.
Cunningham~\cite{Cun85} generalized the concept of strength
to edge-weighted graphs and proposed a strongly polynomial-time 
algorithm to compute it. 
Computing the strength of a graph is a special type of ratio
optimization in the field of submodular function
minimization~\cite{Fuj05}.
Cunningham used the strength of a graph to address 
two different one-player optimization problems: 
the optimal attack and reinforcement of a network.
%
%
The prime-partition we use is a truncated version of the 
principal-partition, 
first introduced by Narayanan~\cite{Nar74} and Tomizawa~\cite{Tom76}.
The principal-partition was used in an extension of 
Cunningham's work to an online setting~\cite{PN00}.
Our work complements that of Cunningham and its successors
by analyzing a new two-player game.
%


The nucleolus of the spanning connectivity game 
can be seen as a special maxmin strategy in the
wiretap game. 
The connection between the nucleolus of a 
cooperative game and equilibrium strategies in a zero-sum
game has been investigated
before in a general context~\cite{NucleolusMatrixGame}.
However, in many cases the nucleolus is hard to compute.
The computational complexity of computing the nucleolus 
has attracted much
attention~\cite{KFK2001}, with both negative results 
\cite{edithcoop,FKK1998,DFS06}, 
%
and positive results
\cite{GMOZ1996,ElkindSODA09,kUIPPERS-CONVEX,ASS1994}.
%
%
Our positive results for the spanning connectivity game
are in contrast to the negative results presented in~\cite{AAIM}, 
where it is shown that the problems of computing the 
Shapley values and Banzhaf values are \#P-complete for the 
spanning connectivity game.
Those results are a strengthening of the hardness results
for the more general, min-base games, introduced
in~\cite{NZKI97}, and the positive results here thus
apply to a special case of those games.

\eat{
The wiretap game is similar in spirit to the zero-sum 
game of~Alon et al.~\cite{AKPW}, which is played in a
edge-weighted graph.
There one player picks an edge and the other a spanning
tree.
The games differ in their payoff function.
There, the application is to the $k$-server problem, but
only bounds are proven.
}

\eat{
\paragraph{Outline of paper.}
In Section~\ref{s-related}, a brief survey of related work 
is provided. 
Section~\ref{s-coopGames} and Section~\ref{s-prereq} 
introduce the cooperative game theory and graph theory 
prerequisites respectively. 
In Section~\ref{s-leastcore}, the supports of the elements 
in the least core are characterized. 
Section~\ref{s-suppNuc} defines the graph decomposition 
and describes a method to compute the support of the 
nucleolus of the connectivity game. 
The graph decomposition gives sufficient information 
to compute the nucleolus. 
This is explained in Section~\ref{s-findNuc}. 
In Section~\ref{s-wiretap}, the connection between the
wiretap game and connectivity game is discussed. 
Finally, Section~\ref{s-conc} presents conclusions 
and open questions.
Moreover, we explain the relationship between 
these variants of the principal-partitions
in detail in Section~\ref{s-partition} 
}

\section{The wiretap game}
\label{s-wiretap}

The strategic form of the wiretap game is 
defined implicitly by the graph $G=(V,E)$.
The pure strategies of the wiretapper are the edges $E$ 
and the pure strategies of the hider are the set of
connected spanning subgraphs $\S$. 
An element of $\S$ is a set of edges, with a typical 
element denoted by $S$. 
The wiretapper receives payoff one if the edge he 
chooses is part of the spanning subgraph chosen by
the hider, and receives payoff zero otherwise.
Thus, the value of the game is the probability that
the wiretapper can secure for wiretapping the connected 
spanning subgraph chosen by the hider.

%
Let $\Delta(A)$ be the set of mixed strategies 
(probability distributions) on a finite set $A$.
By the well-known minmax theorem for finite zero-sum games,
the wiretap game~$\Gamma(G)$ has a unique {\em value}, 
defined by
\begin{equation}
\label{e-maxminMinmax}
val(\Gamma)= \max_{x\in\Delta(E)}\min_{S\in \S} \sum_{e\in
S} x_e =
\min_{y\in\Delta(\S)}\max_{e\in E} \sum_{\{S\in
\S: e \in S\}} y_S\ .
\end{equation} 
%
The equilibrium or {\em maxmin} strategies of the wiretapper 
are
the solutions
$\{ x\in \Delta(E)\ 
|\ \sum_{e\in S} x_e \ge val(\Gamma)\ \text{for all}\ S\in
\S\}$ 
%
to the following linear program,
which has the optimal value $val(\Gamma)$.
\begin{equation}
\label{e-LPmaxmin}
\begin{array}{ll}
\max & z \\
\text{s.t.} & \sum_{e\in S} x_e \ge z\ 
\text{for all}\ S\in \S\ ,\\
& x \in \Delta(E)\ .\\
\end{array}
\end{equation}
Playing any maxmin strategy guarantees the 
wiretapper
a probability of successful wiretapping of at 
least $val(\Gamma)$. 
The equilibrium or {\em minmax} strategies of the 
hider 
are 
$\{ y\in \Delta(\S)\ |\ \sum_{\{S\in
\S : e \in S \}} y_S \le val(\Gamma)\ \text{for all}\
e\in E\}$.
Playing any minmax strategy guarantees the hider to 
suffer a probability of successful wiretapping of no more 
than $val(\Gamma)$. 
The following simple observation shows the importance of
minimum connected spanning graphs in the analysis of the wiretap
game.
For a mixed strategy $x\in \Delta(E)$ and pure strategy 
$S\in \S$, the resulting probability of a successful wiretap 
is $\sum_{e\in S} x_e$.
We denote by $G^x$ the edge-weighted graph comprising the 
graph $G$ with edge weights $x(e)$ for all $e\in E$.
Let $w^*(x)$ be the weight of a minimum connected spanning
graph of $G^x$.

\begin{fact}
\label{fact:best-response}
The set of pure best responses of the hider against 
the mixed strategy $x\in \Delta(E)$ is 
$$
\{S\in \S \ |\ \sum_{e\in S} x_e = w^*(x) \}\ .
$$
\end{fact}

We could define the wiretap game by only allowing the hider to pick
spanning trees, however, 
our definition with connected spanning subgraphs
allows a clean connection to the spanning connectivity game.

\section{Overview of results}
\label{sec:overview}

In this section, we present our results. 
We start with the basic notations and definitions.
From here on we fix a connected graph $G=(V,E)$. 
Unless mentioned explicitly otherwise, 
any implicit reference to a graph is to $G$ and
 $\a$ is an \dist, which is a probability
distribution on the edges $E$. 
For ease, we often refer to the weighted graph $G^\a$ 
simply by $\a$, where this usage is unambiguous.
For a subgraph $H$ of $G$, we denote by $\a(H)$ the sum
$\sum_{e\in E(H)} \a(e)$, where $E(H)$ is the edge set of
$H$.
We refer to equilibrium strategies of the wiretapper as 
\mdists.

\begin{definition}~\label{def:alpha-partition}
For every \dist~$\a$, we denote its
distinct weights by $x_1^\a>\ldots>x_m^\a\geq 0$ 
and define $\mathcal{E}({\a}) = \{E_1^{\a},\dots,E_m^{\a}\}$
such that $E_i^{\a} = \{ e\in E\ |\ \a(e) = x_i^\a\}$
for $i=1,\dots,m$.
\end{definition}


Our initial goal is to characterize those partitions
$\E(\a)$ that can arise from \mdists $\a$.
We start with the following simple setting.
Assume that the wiretapper is restricted to choosing a strategy $\a$
such that $|\E(\a)| =2$, and $x^\a_2 = 0$.
Thus, the wiretapper's only freedom is the choice of the set $E_1^\a$.
What is his best possible choice?
By Fact~\ref{fact:best-response}, a best response against 
$\a$ is a \mcsg $H$ of $\a$.
So the wiretapper should choose $E_1^\a$ so as to maximize
$\a(H)$.
How can such an $E_1^\a$ be found?
To answer, we relate the weight of a \mcsg $H$ of $\a$ to $E_1^\a$.

To determine $\a(H)$, we may assume about $H$ that for every 
connected component $C$ of $(V,E\setminus E_1^\a)$
we have $E(H)\cap E(C) = E(C)$, since $\a(e) = 0$ for every $e \in E(C)$. 
We can also assume that $|E_1^\a\cap E(H)|$ is the
number of connected components in $(V,E\setminus E_1^\a)$ minus $1$,
since this is the minimum number of edges in $E(H)$ that 
a \csg may have.
To formalize this we use the following notation. 

\begin{definition}
Let $E'\subseteq E$. We set $C_G(E')$, to be the number of 
connected components in the graph $G\setminus E'$,
where $G\setminus E'$ is a shorthand for $(V,E\setminus E')$.
If $E' = \emptyset$ we 
just write $C_G$.
\end{definition}

Using the above notation,  
 a \csg  $H$ is a \mcsg of $\a$ 
if 
$
|H\cap E^\a_1| = C_G(E^\a_1)-C_G = C_G(E^\a_1)-1.
$
Now we can compute $\a(H)$.
By definition, $x^\a_1 = \frac{1}{|E^\a_1|}$  and 
$x^\a_2 = 0$ and therefore
$$
\a(H) = \frac{C_G(E^\a_1)-C_G}{|E^\a_1|}.
$$ 
We call this ratio that determines $\a(H)$ the cut-rate of $E^\a_1$.
Note that it uniquely determines the weight of a \mcsg of $\a$.
 
\begin{definition}~\label{def:set-cut-rate}
Let $E'\subseteq E$.
The {\em \crate} of $E'$ in $G$ is denoted by $\CR_G(E')$ and defined
as follows.
\begin{equation}
\CR_G(E') \coloneqq
\begin{cases}
 \frac{C_G(E')-\cge}{|E'|}\quad &\text{if
} |V|>1 \text{ and } |E'| > 0\ ,\\
0 \quad &\text{otherwise}\ .
\end{cases}
\end{equation}
We write $\CR(E')$, unless we make a
point of referring to a different graph.
\end{definition}

Thus, when $|\E(\a)| = 2$ and $x_2^\a = 0$, 
a best choice of $E^\a_1$ is one for which
$\CR(E^\a_1)$ is maximum.
Since $E$ is finite, an $E^\a_1$ that maximizes 
$\CR(E^\a_1)$ exists.
 
\begin{definition}~\label{def:opt}
The \emph{\crate} of $G$ is defined as
$
\opt \coloneqq \max_{E'\subseteq E} \CR(E')\ .
$
\end{definition}

By $opt$, we always refer to the \crate of the graph $G$. 
In case we refer to the cut-rate of some other graph, 
we add the name of the graph as a subscript.
The value $\opt$ is a well known and studied attribute of a graph.
It is equal to the reciprocal of the strength of a
graph, as defined by Gusfield~\cite{Gus83} and named by Cunningham
\cite{Cun85}.
There exists a combinatorial algorithm for computing 
the strength, and hence $\opt$, 
 that runs in time polynomial in the {\em size} of the graph,
by which we always mean $|V|+|E|$.

We generalize the above technique to the case that $\a$ 
is not restricted.
Assume again that $H$ is a \mcsg of $\a$.
Intuitively, even if $\a$ has more than $2$ distinct weights 
we would expect $|E_1^\a \cap E(H)|$ to be as small as possible, 
i.e., $C_G(E^\a_1)-C_G$.
We would also expect $|(E_1^\a\cup E_2^\a) \cap E(H)|$ to be
as small as possible, i.e., $C_G(E^\a_1\cup E_2^\a)-C_G$.
If these both hold then 
$|E_2^\a \cap E(H)| = C_G(E^\a_1\cup E_2^\a)-C_G(E^\a_1)$,
which is the increase in the number of components
we get by removing the edges of $E_2^\a$ from $G\setminus E^\a_1$. 
Thus, the total weight contributed to $H$ by edges in $E(H)\cap E(E_2^\a)$
is $x_2^\a(C_G(E_1^\a  \cup E_2^\a) - C_G(E_1^\a))$.
Now, unlike the previous case, we do not know $x_2^\a$.
However, this is not a problem since, as we shall see, 
we are interested in the ratio 
\eat{
$\frac{\a(E(H)\cap E_2^\a)}{\a(E_2^\a)}$,
which is $\frac{C_G(E_2^\a) - C_{G}(E_1^\a)}{|E_2^\a|}$.
}
$$
\frac{\a(E(H)\cap E_2^\a)}{\a(E_2^\a)} = \frac{C_G(E_1^\a
\cup E_2^\a) -
C_{G}(E_1^\a)}{|E_2^\a|}\ .
$$
We use the following notation to express this and its
extension to more weights.

\begin{definition}~\label{def:cut-rate}
For $\ell=1,\dots,|\E(\a)|$ we set 
$$\CR_\ell^\a = 
\frac{C_G(\cup_{i=1}^{\ell}E_i^\a) - C_{G}(\cup_{i=1}^{\ell-1}E_i^\a)}
{|E_\ell^\a|}.$$
\end{definition}

The intuition above indeed holds, as stated in the following 
proposition, which we prove in
Appendix~\ref{Section-Preliminaries}. 

\begin{proposition}~\label{prop:mcsg}
Let  $H$ be a \mcsg of $\a$.
Then $|E(H) \cap E_\ell^\a| = |E_\ell^\a|  cr_\ell^\a$
for every $\ell$ such that $x_\ell^\a > 0$.
\end{proposition}

Using Proposition~\ref{prop:mcsg} we can 
relate the weight of a \mcsg of $\a$ to the sets of $\E(\a)$.
This relationship also characterizes the \mdists, which are the 
\dists whose \mcsg weight is the maximum possible.

\begin{theorem}~\label{thm:strong-dist}
Let $H$ be a \mcsg of $\a$ and $m = |\E(\a)|$.
Then $\a(H) \leq opt$ and we have $\a(H) = opt$ if and only if
\begin{enumerate}
\itemsep2mm
\item
$cr_\ell^\a = opt$ for $\ell = 1,\dots,m-1$, and 
\item
if $cr_m^\a \ne opt$ then $x_m^\a=0$.
\end{enumerate}
\end{theorem}

Theorem~\ref{thm:strong-dist} is proved in
Appendix~\ref{sec:strong-dist}.
An immediate implication of Theorem~\ref{thm:strong-dist} is that 
$\opt$ is an upper bound on the value the wiretapper can achieve.
This also follows from the well-known fact that the
fractional packing number of spanning trees of a graph is
equal to the strength of a graph, which in turn follows from the
theorems of Nash-Williams~\cite{NashW} and
Tutte~\cite{Tutte} on the integral packing number (see also 
\cite{CMV06}).
Since we have already seen that indeed the wiretapper can
achieve $\opt$ by distributing all probability mass equally over an
edge set that has \crate $opt$, we get the following.

\begin{corollary}~\label{cor:value-of-game}
The value of the wiretap game is $\opt$. 
\end{corollary}

We know what the value of the game is and we know 
a characterization of the~$\E(\a)$'s for \mdists $\a$.
Yet this characterization does not give us a simple way to
find \mdists.
Resolving this is our next goal.
Since the set of \mdists is convex, it is easy to show 
that there exists a \mdist $\be$ such that for every $e_1,e_2\in E$ 
we have $\be(e_1) = \be(e_2)$ if and only if $\gamma(e_1) = \gamma(e_2)$ 
for every \mdist $\gamma$.
This implies that $\E(\be)$ refines $\E(\gamma)$ for every \mdist $\gamma$, 
where by ``refines'' we mean the following.

\begin{definition}~\label{def:Prefinement}
Let $\mathcal{E}_1,\mathcal{E}_2$ be partitions of $E$.
Then $\mathcal{E}_1$ \emph{refines} $\mathcal{E}_2$ if for every set 
$E'\in \mathcal{E}_1$ there exists a set $E''\in \mathcal{E}_2$ such that
$E'\subseteq E''$.
\end{definition}

Thus, 
there exists a partition of $E$ that is equal to 
$\E(\be)$ for some \mdist $\be$ and refines $\E(\gamma)$ for
every \mdist~$\gamma$.
We call such a partition the {\em prime-partition}. 
It is unique since there can not be different partitions 
that refine each other.

\begin{definition}~\label{def:prime-partition}
The \emph{prime-partition} $\pp$ is the unique partition 
that is equal to $\E(\be)$ for some \mdist $\be$ and
 refines $\E(\gamma)$ for every \mdist $\gamma$.
\end{definition}

\begin{theorem}~\label{thm:unique-prime-partition}
The prime-partition exists and can be computed in time polynomial in the size of $G$.
\end{theorem}

Theorem~\ref{thm:unique-prime-partition} is proved in 
Appendix~\ref{sec:prime-partition}.
The prime-partition $\pp$ reveals a lot about the structure of the 
\mdists. Yet by itself $\pp$ does not give us a simple means for
generating \mdists. 
Using the algorithm for finding $\pp$ one can show that, 
depending on $G$, there may be a unique element
in $\pp$ whose edges are assigned $0$ by every \mdist.


\begin{lemma}~\label{lem:deg-set}
$cr_G(E) \ne opt$ if and only if 
there exists a unique set $D\in \pp$ such that 
for every \mdist $\a$ and $e\in D$ we have $\a(e)=0$.
If $D$ exists then it can be found in time polynomial in the size
of $G$.  
\end{lemma}

Lemma~\ref{lem:deg-set} is proved in
Appendix~\ref{sec:deg-set}.
From here on we shall always refer to the set $D$ in
Lemma~\ref{lem:deg-set} as the {\em degenerate set}.
For convenience, if $D$ does not exist then we shall treat 
both $\{D\}$ and $D$ as the empty set.
See Figure~\ref{figure} for an example of the
prime-partition and the degenerate set.

We use the prime-partition to define a special subset of 
the \mcsgs that we call the {\em \ocsgs}, which are useful for 
proving the characterization of \mdists and their
refinements.

\begin{definition}~\label{def:ocsg}
A \csg $H$ is an {\em \ocsg} if for every 
$P\in \pp\setminus \{D\}$ we have
$$|E(H)\cap P| = |P|\cdot \opt\ .$$ 
\end{definition}

\begin{proposition}~\label{prop:ocsg}
There exists an \ocsg.
\end{proposition}

\begin{proof}
Let $\be$ be a \mdist such that $\E(\be)=\pp$.
Let $H$ be a \mcsg of $\be$.
Then by Proposition~\ref{prop:mcsg}, we have that $H$ is an \ocsg.
\qed
\end{proof}

The \ocsgs are the set of the hider's pure strategies that are best 
responses against every \mdist.
 
\begin{proposition}
\label{prop:omni-opt}
For every \dist $\a$ such that $\pp$ refines $\E(\a)$ and
$\a(e)=0$ for every $e\in D$ and \ocsg $H$, we have $\a(H) = \opt$.
\end{proposition}

We prove Proposition~\ref{prop:omni-opt} in
Appendix~\ref{sec:omni-opt}.
The importance of \ocsgs stems from the following scenario.
Assume that $\pp$ refines~$\E(\a)$ and $\a(e)=0$ for every $e\in D$, 
and let $H$ be an \ocsg. 
By Proposition~\ref{prop:omni-opt}, we know that $\a(H)=\opt$.
Suppose we can remove from $H$ an edge from $E(H)\cap P$, where
$P$ is a nondegenerate element of $\pp$,
and add a new edge from another set $P'\setminus E(H)$
in order to get a new \csg.
Assume $\a$ assigns to the edge removed strictly more weight than 
it assigns to the edge added.
Then the new \csg has weight strictly less than $\a(H)$ 
and hence strictly less than $\opt$, since $\a(H)=\opt$ 
by Proposition~\ref{prop:omni-opt}.
Consequently, $\a$ is not a \mdist and we can
conclude that any \dist $\be$ that assigns 
to each edge in $P$ strictly more weight than to the edges
in $P'$ is not a \mdist.
This intuition is captured by the following definition,
which leads to the characterization of \mdists in
Theorem~\ref{thm:prime-order}.

\begin{definition}~\label{def:leadsto}
Let $P,P'\in\pp\setminus \{D\}$ be distinct. 
Then $P$ \emph{leads to} $P'$ if and only if there exists 
an \ocsg $H$ with $e\in P\setminus E(H)$ and $e'\in P'\cap E(H)$ 
such that $(H\setminus \{e'\})\cup \{e\}$ is a \csg.
We denote the ``leads to'' relation by $\mathcal{R}$.
\end{definition}

\begin{definition}~\label{def:agree-induced-order}
An \dist $\a$ \emph{agrees} with $\mathcal{R}$ 
if $\pp$ refines $\E(\a)$ and for every 
$P\in \pp\setminus \{D\}$ that is a parent of $P'\in
\pp\setminus \{D\}$ and $e\in P$, $e'\in P'$ we have $\a(e) \geq \a(e')$, and 
for every $e\in D$ we have $\a(e) = 0$. \end{definition}

\begin{theorem}~\label{thm:prime-order}
An \dist $\a$ is a \mdist if and only if it agrees
with~$\mathcal{R}$.
\end{theorem}

Theorem~\ref{thm:prime-order} is proved in
Appendix~\ref{sec:prime-order}.
By definition, there exists a \mdist $\be$ with
$\E(\be) = \pp$.
By Theorem~\ref{thm:prime-order}, we have that $\be$ 
agrees with~$\mathcal{R}$ and hence the following holds.

\begin{proposition}
The relation $\mathcal{R}$ is acyclic.
\end{proposition}



This allows us to define the acyclic parent-child relation, which is
a simplification of~$\mathcal{R}$ and easy to find.

\begin{definition}~\label{def:parent}
Let $P,P'\in\pp\setminus \{D\}$ be distinct. 
We say that $P$ is a \emph{parent} of $P'$ 
(conversely $P'$ a \emph{child} of $P$) if $P$ leads to $P'$ and
there is no $P''\in\pp$ such that $P$ leads
to $P''$ and $P''$ leads to $P'$.
We refer to the relation as the parent-child relation and denote it by~$\po$.
\end{definition}


The following is an immediate corollary of
Theorem~\ref{thm:prime-order} and
Definition~\ref{def:parent}.

\begin{corollary}~\label{cor:prime-order}
An \dist $\a$ is a \mdist if and only if it agrees with~$\po$.
\end{corollary}

See Figure~\ref{figure} for an example of an \ocsg and the
exchangeability of edges between a parent and child.
Corollary~\ref{cor:prime-order} defines a linear inequality for each
parent and child in the relation~$\po$.
Along with the inequalities that define a probability
distribution on edges, this gives a small number of
two-variable inequalities describing the \maxminpoly.
%
In Appendix~\ref{s-extreme} we characterize the extreme
points of \maxminpoly.
%
The proof of the following theorem, which states that $\po$
can be found in polynomial time, can be found in 
Appendix~\ref{sec:compute-po}.

\begin{theorem}~\label{thm:compute-po}
The parent-child relation $\po$ can be computed in time polynomial 
in the size of $G$. 
\end{theorem}


The wiretapper will in general have a choice of infinitely
many \mdists.
To choose a \mdist, it is natural to consider refinements of
the Nash equilibrium property that are beneficial to the 
wiretapper if the hider does not play optimally.
First we show how to minimize the number of pure best
responses of the hider.
To do this, we use the relation $\po$ to characterize
a special type of \mdist which achieves this.
We call this a \pdist.
The \pdists are characterized by the following lemma.

\begin{definition}~\label{def:mdist}
A \mdist $\a$ is a {\em \pdist} if the number of the hider's pure best 
responses against it is the minimum possible. 
\end{definition}

\begin{lemma}~\label{lem:prime-dist}
An \dist $\ga$ is a \pdist if 
and only if $\ga(e)>0$ for every $e\in E\setminus D$,
and for every $P,P'\in \pp\setminus \{D\}$ such that
$P$ is a parent of $P'$ and every $e\in P$, $e'\in P'$,
we have $\ga(e') > \ga(e'')$.
\end{lemma}

Using this characterization one can easily check whether 
$\a$ is a \pdist and one can also easily construct a \pdist.

We prove Lemma~\ref{lem:prime-dist} in Appendix~\ref{sec:prime-dist}.
The proof runs as follows.
First we show that for any $\a$ that satisfies the condition
of the lemma, every \mcsg is an \ocsg.
Hence, using Proposition~\ref{prop:omni-opt}, we get that
for any $\a$ that satisfies the condition
of the lemma, a \csg is a \mcsg of~$\a$ if and only if it 
is an \ocsg.
These are the only such \mdists, since any \mdist that does not
satisfy the  condition of the lemma has a parent and its child
whose edges get the same weight.
Consequently, by the definition of parent and child, it 
has a \mcsg that is not an \ocsg.

\eat{
We know that every \ocsg is a pure best response against
every \mdist.
The following lemma shows that for \pdists these are the
only pure best responses.

\begin{lemma}~\label{lem:prime-CSG}
If $\alpha$ is a \pdist, then every minimum \csg of $\alpha$
is an \ocsg. 
\end{lemma}
}

%
%
\pgfdeclarelayer{background}
\pgfsetlayers{background,main}

\renewcommand\r[1]{{{#1}}}
\renewcommand\b[1]{{{#1}}}

\begin{figure}[htb]
\centering

\begin{tikzpicture}[auto,style=thick]

\tikzstyle{vertex}=[circle,fill=black,minimum size=7pt,inner sep=0pt]

\tikzstyle{selected edge} = [draw,line width=5pt,-,red!50]
\tikzstyle{blue edge} = [draw,line width=5pt,-,blue!50]
\tikzstyle{green edge} = [draw,line width=5pt,-,green!50]
\tikzstyle{yellow edge} = [draw,line width=5pt,-,yellow!50]

\foreach \name/\angle/\text in
{Q-1/234/,Q-2/162/,Q-3/90/,Q-4/18/,Q-5/-54/}
  \node[vertex,xshift=4cm,yshift=2cm] (\name) at
(\angle:1cm) {$\text$};

\foreach \from/\to in {1/2,2/3,3/4,4/5,5/1,1/3,2/4,3/5,4/1,5/2}
    { \draw[dotted] (Q-\from) -- (Q-\to);}

\foreach \name/\angle/\text in
{P-2/180/,P-3/90/,P-4/0/,P-5/-90/}
  \node[vertex,xshift=2cm,yshift=0cm] (\name) at
(\angle:1cm) {$\text$};

\foreach \name/\angle/\text in 
{R-2/180/,R-3/90/,R-4/0/,R-5/-90/}
    \node[vertex,xshift=6cm,yshift=0cm] (\name) at
(\angle:1cm) {$\text$};

\foreach \from/\to in {2/3,3/4,4/5,2/4,3/5,5/2}
    { \draw[dotted] (P-\from) -- (P-\to); \draw[dotted] (R-\from) -- (R-\to);}

\foreach \from/\to in {1/2,1/5,4/5,3/4}
    { \draw[style=ultra thick] (Q-\from) -- (Q-\to);}

\foreach \from/\to in {3/4,4/5,5/2}
    { \draw[style=ultra thick] (R-\from) -- (R-\to);}

\foreach \from/\to in {3/2,2/5,5/4}
    { \draw[style=ultra thick] (P-\from) -- (P-\to);}

\begin{pgfonlayer}{background}
        \foreach \source / \dest in {P-3/Q-1,Q-5/R-3}
        \path[selected edge] (\source.center) -- (\dest.center);
        \foreach \source / \dest in {P-4/R-2,P-5/R-5}
        \path[blue edge] (\source.center) -- (\dest.center);
        \foreach \from/\to in {2/3,3/4,4/5,2/4,3/5,5/2}
        {
            \path[green edge] (P-\from.center) -- (P-\to.center);
            \path[yellow edge] (R-\from.center) -- (R-\to.center);
        } 
\end{pgfonlayer}

\path[draw,style=ultra thick] (Q-5) --  (R-3);
\path[draw,dotted] (P-4) --  (R-2);
\path[draw,dotted] (P-5) --  (R-5);
\path[draw,dotted] (P-3) --  (Q-1);

\eat{
\path[] (P-4) -- node {$\b{e_2}$} (R-2); 
\path[] (P-3) -- node {$\r{e_1}$} (Q-1);
}

\path[] (P-3) --  node {\b{$e_1$}} (Q-1);
\path[] (Q-5) --  node {\b{$e_2$}} (R-3);
\path[] (P-4) --  node {\b{$e_3$}} (R-2);
\path[] (P-5) --  node {\b{$e_4$}} (R-5);

{
\path[draw,style=ultra thick] (P-4) -- (R-2); 
}



\tikzstyle{layer}=[]
\tikzstyle{edge} = [draw,thick,->]

\tikzstyle{mysplit}=[rectangle split,draw, rectangle split parts=2]

\tikzstyle{rect}=[rectangle,draw]

\node[rect, minimum size=20pt, fill={red!50}]   (12) at
(10,2.5) {\footnotesize $E_1$};
\node[rect, minimum size=20pt, fill={blue!50}]  (22) at
(10,1) {\footnotesize $E_2$};
\node[rect, minimum size=20pt, fill={green!50}] (21) at
(9,-0.5) {\footnotesize $E_3$};
\node[rect, minimum size=20pt, fill={yellow!50}](32) at
(11,-0.5) {\footnotesize $E_4$};
\eat{
\node[rect, minimum size=20pt, ]                (33) at
(6,-1.5) {\footnotesize $E_5$};
}

\path[edge] (12) -- (22);

\foreach \source/ \dest in {22/21,22/32}
    \path[edge] (\source) -- (\dest);

\foreach \pos/\name/\nameA in 
{{(12,2.5)/d2/L_3},{(12,1)/d1/L_2}, {(12,-0.5)/d0/L_1}},
        \node[layer] (\name) at \pos {$\nameA$};

\end{tikzpicture}

\caption{
{\sc Left:}
The left figure illustrates the prime-partition $\pp =
\{E_1,\ldots, E_5\}$.
For this graph, $opt=1/2$.
The set $E_1=\{e_1,e_2\}$, the
set $E_2=\{e_3,e_4\}$, the set $E_3$ is equal to the edges
of the left $K_4$, the set $E_4$ is equal to the edges of
the right $K_4$, and the set $E_5$ is equal to the edges of
the $K_5$. 
Suppose that \mdist $\beta$ is such that $\E(\beta)=\pp$, 
and $E_i^\beta = E_i$ for $i=1,\ldots,5$.
(There will be other \mdists with the same partition in which 
$E_3$ and $E_4$ exchange roles.)
Removing $E_1$ from the graph
creates one extra component by removing two edges, so 
we have $\CR_1^\beta = \CR(E_1)=opt=1/2$.
Similarly we have $\CR_k^\beta=1/2$ for all $k=1,\ldots 4$.
However, $\CR_5^\beta=4/10<1/2$ and so the set $E_5$ 
is a degenerate set, as per Lemma~\ref{lem:deg-set}.
The figure shows the subgraph $H$ indicated with
solid edges. It is an \ocsg, using two edges from each of
the $K_4$'s, one edge from the two edges that connect the
two $K_4$'s, and one edge from the two edges that connect
the two $K_4$'s to the $K_5$.
Within the $K_5$, an \ocsg~can use more than four edges, as
this $K_5$ corresponds to the final element of the 
prime-partition with any strong linear order and achieves cut-rate
$4/10$, which is worse than $opt=1/2$.
The edge $e_3$ can be replaced with the edge $e_1$. 
Thus, the edges in the element of the prime-partition containing
$e_1$ must have weight at least that of the edges in the
element of the prime-partition containing $e_3$.
{\sc Right:} The right figure illustrates the partial order~$\po$ 
and its layers $\{L_1,L_2,L_3\}$.
}
\label{figure}
\end{figure}

%

We have already seen how to minimize the number of pure best
responses of the hider, by playing a \pdist.
We now show how to uniquely maximize the weight of a pure
second-best response by choosing between \pdists. 
This maximizes the least punishment that the hider will incur 
for picking a non-optimal pure strategy.

Against a \pdist, the candidates 
for pure second-best responses are those \csgs
that differ from \ocsgs in at most two edges.
For each parent and child we have at least one of these
second-best responses.
A second-best response either is a best response with 
one extra edge, or it differs from a best response in two edges, 
where it has one less edge in a child of $\po$ 
and one more in the child's parent.

We are only interested in the case that $\opt<1$, since 
the graph has $\opt= 1$ if and only if it contains a bridge,
in which case the value of the game is one and 
the hider does not have a second-best response. 
From here on we assume the following.
\begin{assumption}
$\opt < 1$. 
\end{assumption}

Intuitively, to maximize the weight of a second-best
response, we want to minimize the number of distinct
weights.
The minimum number of distinct positive weights 
we can achieve for a \pdist is equal to the number of
elements in the longest chain in the parent-child relation. 
This motivates the following definition. 

\begin{definition}~\label{def:layers}
We define $\mathcal{L}_1,\mathcal{L}_2,\dots$ 
inductively as follows.
The set $\mathcal{L}_1$ is all the sinks of $\po$ excluding $D$.
For $j=2,\dots$, we have that 
$\mathcal{L}_j$ is the set of all the sinks 
when all elements of
$\{D\}\cup(\cup_{i=1,\ldots,j-1}\mathcal{L}_i)$ 
have been removed from $\po$.
\end{definition}
Note that $\po$ is defined only over nondegenerate elements of $\pp$
and hence the degenerate set is not contained in any of $\mathcal{L}_1,\mathcal{L}_2,\dots$.

\begin{definition}~\label{def:prime-layers} 
The \emph{layers} $\mathcal{L}= \{L_1,\dots,L_t\}$ of $G$ are
 $L_i = \cup_{E'\in \mathcal{L}_i} E'$ for $i = 1,\ldots,t$.
\end{definition}

See Figure~\ref{figure} for an example of layers.
The following theorem shows that there is a unique \mdist
that maximizes the difference between the payoff of a best
and second-best response.
This unique \mdist turns out to be the nucleolus of the
spanning connectivity game, as explained in
Section~\ref{s-scg}.
For convenience, 
we refer to this 
strategy as the nucleolus.

\begin{theorem}~\label{thm:Nucleolus}
Let
\[\kappa = \frac{1}{\sum_{i=1}^{t} i \cdot |L_i|}\ .\]
The nucleolus $\nu$ has $\nu(e) = i\cdot \kappa$
for every $i\in \{1,\dots,t\}$ and $e \in L_i$ and
$\nu(e) = 0$ otherwise.
\end{theorem}

Theorem~\ref{thm:Nucleolus} is proved in
Appendix~\ref{sec:Nucleolus}.
The proof says that the weight of a second-best response is
$\opt + \kappa$, and this must be optimal since all the
weights are  multiples of $\kappa$.
For all other \pdists 
there is
a 
second-best response with weight less than $\opt + \kappa$.



\section{Spanning connectivity games}

\label{s-scg}

A \emph{simple cooperative game} $(N,v)$ consists of a player set
$N=\{1,\ldots,n\}$ and characteristic function 
$v:2^N \rightarrow \{ 0,1 \}$ with $v(\emptyset)=0$, 
$v(N)=1$, and $v(S)\leq v(T)$ whenever $S \subseteq T$. 
A coalition $S \subseteq N$ is \emph{winning} 
if $v(S)=1$ and \emph{losing} if $v(S)=0$.  
The payoff vector to the players $x=(x_1,\ldots,x_n)$ 
satisfies $x(N)=v(N)=1$, with $x(S)=\sum_{i \in S} x_i$.
It is called an \emph{imputation} if $x_i\geq v(\{i\})$ 
for all $i\in N$.
%
%
For a game $(N,v)$ and imputation $x=(x_1,...,x_n)$, 
the \emph{excess} $e(x,S)$ of a coalition 
$S$ under $x$ is defined as 
$
e(x,S)= x(S)-v(S).
$

We relate our analysis of the wiretap game to two 
cooperative solutions based on the excess of
coalitions: the \emph{least core} and the \emph{nucleolus}, 
which is a unique point in the least core.
%
An imputation $x$ is in the \emph{$\epsilon$-core} if 
$e(x,S)\geq -\epsilon$ for all $S \subset N$.
An imputation $x$ is in the \emph{least core} if it is 
in the $\epsilon$-core for the smallest possible~$\epsilon$. 
The \emph{excess vector} of an imputation $x$ is 
$
(e(x,S_1),...,e(x,S_{2^n}))\ ,
$
where $e(x,S_1)\leq e(x,S_2) \leq \cdots \leq e(x,S_{2^n})$. 
The \emph{nucleolus} is the element of the least core 
which has the largest excess vector lexicographically.
The nucleolus is unique~\cite{1969Nucleolus}.
We denote the {\em distinct} excesses by 
$(-\epsilon_1,\ldots,-\epsilon_t)$ for $t\le
2^n$, where $\epsilon_1>\cdots>\epsilon_t$.
Note that, by definition, $\epsilon_1 = \epsilon$, the least
core value, which is the optimal value of the objective
function in~\eqref{e-LP}.

For 
a graph $G=(V,E)$ with at least three nodes, 
the \emph{spanning connectivity game}, 
introduced in \cite{AAIM}, 
has player set $E$ and characteristic function 
$$
v(S)= \left \{ \begin{array}{ll} 1, & \textrm{if there
exists a spanning tree $T=(N,E')$ such that $E'\subseteq
S$}\ ,\\
0, & \textrm{otherwise}\ .\end{array}\right.$$
%
Since the graph has at least
three nodes, $v(\{i\}) = 0$ for all $i \in E$,
so $x$ is an imputation if and only if it is a 
probability distribution on players, i.e., $x \in \Delta(E)$. 
The least core of the spanning connectivity game is the set of 
all solutions to the following linear program:
\begin{equation}
\label{e-LP}
\begin{array}{ll}
\min & \epsilon  \\
\text{s.t.} & e(x,S) \ge -\epsilon\ ,\ \text{for all}\
S\subset E\ ,\\
& x \in \Delta(E)\ .\\
\end{array}
\end{equation}
First we show that the least core is identical to the
\maxminpoly. 

\begin{proposition}
\label{l-leastcoreMaxmin}
The least core of the spanning connectivity
game is the set of \mdists of the wiretap game.
\end{proposition}

\begin{proof}
The problems of finding a \mdist and an element of the least 
core are given by the LPs~\eqref{e-LPmaxmin} (from
Section~\ref{s-wiretap}) and~\eqref{e-LP}, respectively.
The solution of~\eqref{e-LP} satisfies $\epsilon\ge 0$,
since we have $x(S)-v(S)\ge -\epsilon$ and $x(S)\le 1$, and 
for any winning coalition $v(S)=1$. 
So, for any losing coalition $S$, where $v(S)=0$, 
the inequality 
$e(x,S)= x(S)-v(S)\ge -\epsilon$ in~\eqref{e-LP}  
is redundant, and only the inequalities for winning coalitions,
i.e., connected spanning subgraphs are needed.
Note that $x(S)=\sum_{e\in S} x_e$.
Hence, the linear programs~\eqref{e-LPmaxmin} and~\eqref{e-LP} 
have the same solutions with $z=1-\epsilon$, except for the 
objective functions that differ only by a constant.
\qed
\end{proof}

Now we show that the nucleolus is the unique \mdist that 
minimizes the number of pure best responses of the hider 
and, given this, maximizes the probability arising from the 
hider playing a pure second-best response.

\begin{proposition}
The nucleolus of the spanning connectivity game is identical
to the \mdist defined in Theorem~\ref{thm:Nucleolus}.
\end{proposition}

\begin{proof}
The \mdist $\nu$ in Theorem~\ref{thm:Nucleolus} minimizes the
number of pure best responses of the hider, i.e., it
minimizes the number of minimum connected spanning subgraphs
in $G^\nu$. This is equivalent to minimizing the number of
$\epsilon_1$-coalitions in the spanning connectivity game.
Moreover, it \emph{uniquely} maximizes the probability for a
successful wiretap of a second-best response of the hider, 
which is equivalent to maximizing $\epsilon_2$.
\qed
\end{proof}



\section{Further research}
\label{s-conc}

The equilibrium strategies of the hider can be found using
the complete refined principal partition~\cite{CGHL92}
(our prime-partition is a truncation of this one).
A characterization of these minmax strategies would be an  
interesting next step.

\vspace*{-2mm}

\paragraph{Extensions to wiretap game.}
There are a number of natural extensions to the wiretap
game.
For example, if the wiretapper is allowed to pick
multiple edges.
In Figure~\ref{figure}, if the wiretapper
can pick two edges, then by choosing $e_1$ and $e_2$, he
guarantees success.
With the number of edges to pick as input, is this problem
computationally tractable, or hard?
One could consider variants where the nodes of the hider are 
a subset of all nodes, unknown to the wiretapper.
Another natural extension is to make the game dynamic with
multiple rounds.

\vspace*{-2mm}

\paragraph{Further equilibrium refinements.}
Potters and Tijs~\cite{NucleolusMatrixGame} define
the ``nucleolus of a matrix game'' and show that for
a matrix the nucleolus, which is no longer unique as
for a cooperative game, corresponds
to the {\em proper equilibria} of the matrix game.
This equilibrium concept, unlike Nash equilibria for 
zero-sum games, is not independent for the two players (for a
Nash equilibrium, one player can independently play 
any maxmin strategy and the other any minmax strategy in
equilibrium).
Unlike the special 
strategy of the wiretapper we compute here, computing a
proper equilibrium will require a
simultaneous analysis of the strategies of both the
wiretapper and hider, however it seems plausible that the
structure we have shown here may be enough to do this.
So, can we efficiently find one proper equilibrium of the wiretap
game, or even a characterization of the complete 
set of proper equilibria?
What about other equilibrium refinements?

\eat{
\paragraph{Cooperative game setting.}
Solution concepts in cooperative games measure 
the importance of players in  the game.
In \cite{AAIM}, it is shown that computing the Banzhaf
values and Shapley-Shubik indices are \#P-hard for the
spanning connectivity game.
Here, we show that computing the nucleolus can be done in
polynomial time, as well as a complete characterization of
all least-core members.
The nucleolus is known to lie in the kernel and bargaining
sets, which are alternative solutions for cooperative games.
Can our analysis be extended to these games?
}

\eat{
\section{Acknowledgements}

We thank Jan van den Heuvel for a number of very useful
pointers.
}

%

\newpage

\small
\bibliographystyle{abbrv}

\newpage

\begin{sloppy}

\setcounter{section}{-1}

\section{Appendix: Preliminaries}
\label{Section-Preliminaries}


{\bf Proof of Proposition~\ref{prop:mcsg}}

Let $H$ be a \mcsg of $\a$.
And let $t$ be the maximum such that $x_t^\a > 0$.
We next show that $|E(H) \cap E_i^\a| = |E_i^\a|  cr_i^\a$
for $i=1,\dots,t$.

Assume for the sake of contradiction that this is not so.
Let $k$ be minimal such that 
$|E(H) \cap E_k^\a| \ne |E_k^\a|  cr_k^\a$.
By the minimality of $k$ we have 
\begin{equation}~\label{equ:mcsg1}
|E(H)\cap(\cup_{i=1}^{k-1}E_i^\a)| = \sum_{i=1}^{k-1} |E_i^\a|  cr_i^\a.
\end{equation}
Set $E' = \cup_{i=1}^k E_i^k$.
By the definition of cut-rate the number of connected components
in $G\setminus E'$ is 
\begin{equation}~\label{equ:mcsg2}
C_G(E') = 1 + \sum_{i=1}^k |E_i^\a|  cr_i^\a.
\end{equation}
Thus $|E(H)\cap E'|$ is at least $\sum_{i=1}^k |E_i^\a|  cr_i^\a$
and therefore by \eqref{equ:mcsg1} we have $|E(H) \cap E_k^\a| \ge |E_k^\a|  cr_k^\a$.

Assume  $|E(H) \cap E_k^\a| > |E_k^\a|  cr_k^\a$.
Then, by \eqref{equ:mcsg1}, we have
\begin{equation}~\label{equ:mcsg3}
 |E(H)\cap E'| > \sum_{i=1}^k |E_i^\a|  cr_i^\a.
\end{equation}
We show next that this implies that there exists a \csg whose weight by 
$\a$ is strictly less than $\a(H)$ in contradiction to $H$
being a \mcsg.
Set $s= C_G(E')$ and let $C_1,\dots,C_s$ be 
the connected components of $G\setminus E'$.
Now as $H$ is a \mcsg the set of edges in $E(H)\cap E'$ does not have a
cycle, otherwise we could have removed one of them to get a
\csg with strictly less weight.
Thus the number of connected components of $E(H)\setminus E'$ is 
$1+|E\cap E(H)|$.
Set $r = |E'\cap E(H)|$ and let $H_1,\dots,H_r$ be the connected components
of $H\setminus E'$.

Note that for each $i\in \{1,\dots,r\}$ there exists a unique $j\in\{1,\dots,s\}$
such that $E(H_i)\subseteq E(C_j)$.
For each $j\in\{1,\dots,s\}$ set $I_j$ to be the set of all $i\in \{1,\dots,r\}$ such that
$E(H_i)\subseteq E(C_j)$.
By \eqref{equ:mcsg2} and \eqref{equ:mcsg3} we have $s<r$
and therefore by the
 pigeon-hole principle there exists $j\in \{1,\dots,r\}$ such that 
$|I_j|>1$.
Since $C_j$ is a connected component and $H$ a \csg there
exist $x,y\in I_j$ and $e=\{u,v\}\in E(C_j)\setminus \cup_{i=1}^{|I_j|}E(H_i)$ such that  $u\in V(H_x)$ and  $v\in V(H_y)$. 
Again because $H$ is a \csg there is a path in $H$ between $u$ and $v$
this path contains edges not in $E(C_j)$ because $u,v$ are in different connected
components of $H\setminus E'$.
Thus this path contains an edge $e'\in E'$ because only 
edges from $E'$ connect the vertices of $C_j$ to the rest of the graph.
Consequently $(H\setminus \{e'\})\cup\{e\}$ is a \csg.
Since $e\not\in E'$ we have $\a(e)< \a(e')$ and consequently
$\a(H) > \a((H\setminus \{e'\})\cup\{e\})$.

\begin{fact}~\label{fact:mcsg}
Let $H$ be a \mcsg of $\a$ and $m=\E(\a)$ then
$\a(H) = \sum_{i = 1}^{m} x_i^\a |E_i^\a| cr_i^\a$
and for each $i=1,\dots,m$ 
if $cr_i^\a<1$ then there exists $e\in E_i^\a \setminus E(H)$.
\end{fact}

\begin{proof}

By definition $|E(H)| = \sum_{i = 1}^{m} |E(H) \cap E_i^\a|$.
Therefore $\a(H) = \sum_{i = 1}^{m} x_i^\a|E(H) \cap E_i^\a|$.
By applying Proposition~\ref{prop:mcsg} we get that
$\a(H)=\sum_{i = 1}^{m} x_i^\a |E_i^\a| cr_i^\a$ .

Fix $i\in\{1,\dots,m\}$. 
By Proposition~\ref{prop:mcsg} we have
 $|E(H) \cap E_{i}^\a| = |E_{i}^\a | cr_i^\a$
and hence if $cr_i^\a<1$ then
 $|E(H) \cap E_{i}^\a| < |E_i^\a |$ and therefore
 $E_i^\a \setminus E(H)$ is not empty.
\qed
\end{proof}

\begin{fact}~\label{fact:composition}
Let $E_1,\dots,E_s\subseteq E$ be such that $E_i\cap E_j = \emptyset$ 
for every distinct $i,j\in \{1,\dots,s\}$.
For $\ell=1,\dots,s$ let $r_\ell$ be the cut-rate of $E_\ell$ 
 in $G\setminus \cup_{i=1}^{\ell-1}E_\ell$.
Assume that $r_\ell \geq y$ ($r_\ell \leq y$) for each $\ell=1,\dots,s$.
Then if there exists $i\in \{1,\dots,s\}$ such that $r_i > y$ ($r_i<y$)
we have $\CR(\cup_{i=1}^{s}E_i) > y$ ($r < y$) and otherwise $\CR(\cup_{i=1}^{s}E_i) =y$.
\end{fact}
\begin{proof}
By the definition of cut-rate 
$C_G(\cup_{i=1}^{s}E_i) = C_G+\sum_{i=1}^s r_i|E_i|$ and hence
$$\CR(\cup_{i=1}^{s}E_i) = \frac{(C_G+\sum_{i=1}^s x_i|E_i|)-C_G}{\sum_{j=1}^s |E_j|}\geq
\frac{\sum_{i=1}^s y|E_i|}{\sum_{j=1}^s |E_j|} = y.$$
Note that the above inequality is strict unless $r_i = y$ for $i=1,\dots,s$.
The proof for the case that $r_i \leq y$ for $i=1,\dots,s$, is the same.
\qed
\end{proof}

\begin{definition}~\label{def:prime-set}
A minimal set $E'\subseteq E$ such that $\CR(E')=\opt$
is a \emph{prime-set} of~$G$.
\end{definition}

\begin{proposition}~\label{prop:prime-set}
For $E',E''\subset E$ such that $\CR(E') = \CR(E'') = \opt$
 the following holds:
\begin{enumerate}
\item $opt_{G\setminus E'}\leq \opt$.
\item If $E'' \ne E'$ then $\CR_{G\setminus E'}(E''\setminus E')=\opt$.
\item If $E''\cap E'\neq \emptyset$ then $\CR(E''\cap E') = \opt$.
\item If $E''\setminus E'\neq \emptyset$ then $opt_{G\setminus E'}=\opt$.
\item If $E'$ is a prime-set then either $E'\subseteq E''$ or 
$E'\cap E'' = \emptyset$.
\end{enumerate}
\end{proposition}
\begin{proof}
Note that $\opt=0$ only if $E=\emptyset$ and therefore in this case the proposition trivially holds.
Assume that $\opt>0$. 
Hence by the definition of cut-rate we have  $E',E''\neq \emptyset$.
We shall also assume that $E'\neq E''$ since otherwise the last four items hold trivially.
We next prove the first item.

Let $E^*\subseteq E\setminus E'$ be such that 
$\CR_{G\setminus E'}(E^*) = opt_{G\setminus E'}$. 
By definition such a set exists.
Observe that $\CR_{G\setminus E'}(E^*) \leq opt$ because otherwise
since $\CR(E') =\opt$ by Fact~\ref{fact:composition}, we
have $\CR(E'\cup E^*) > \opt$, 
which is a contradiction to the maximality of $\opt$.

We now prove the second and third items.
If $E''\cap E'= \emptyset$ then both items trivially hold.
Assume $E''\cap E'\ne \emptyset$.
According to the first item $\CR_{G\setminus E'}(E''\setminus E') \leq \opt$
and by definition $\CR(E''\cap E') \leq \opt$.
Thus as $\CR(E'')= \opt$ by Fact~\ref{fact:composition} 
we have $\CR_{G\setminus E'}(E''\setminus E') = \CR(E''\cap E') = \opt$.

Finally we prove the last two items.
Assume $E''\setminus E'\neq \emptyset$.
By the first item $opt_{G\setminus E'} \le opt$.
By the second item $\CR_{G\setminus E'}(E''\setminus E') = \opt$ and hence
also $opt_{G\setminus E'} \ge opt$ and consequently $opt_{G\setminus E'} = opt$.

Assume that $E'$ is a prime-set.
If $E'\cap E'' \ne \emptyset$ then
by the second item $\CR(E'\cap E'')=\opt$ and hence by the
definition of prime-set $E'\cap E'' = E'$ which implies $E'\subseteq E''$.
\qed
\end{proof}

\section{Appendix: Proof of Theorem~\ref{thm:strong-dist}}~\label{sec:strong-dist}

Let $\be$ be a \dist and $s=|\E(\be)|$.
We say $\be$ is \emph{strong} if 
$cr_\ell^\be = opt$ for $\ell = 1,\dots,s-1$ and
if $cr_s^\be \ne opt$ then $x_s^\be=0$.
From here on in this section $H$ is a \mcsg of $\a$.
Assume $\a$ is strong.
By Fact~\ref{fact:mcsg}
$$\a(H) = \sum_{\ell=1}^{|\mathcal{E}({\a})|} x_\ell^\a |E_\ell^\a| cr_\ell^\a.$$
Therefore as we have $cr_i^\a = \opt$ for every $i$ such that
$x_i^\a>0$ we conclude 
$$
\a(H) = \opt\sum_{\ell=1}^{|\mathcal{E}({\a})|}x_\ell^\a |E_\ell^\a|.
$$
Finally, since $\a$ is an \dist
$\sum_{\ell=1}^{|\mathcal{E}({\a})|}x_\ell^\a |E_\ell^\a| =
1$, we get that $\a(H) = \opt$. 
Now the theorem directly follows from the subsequent lemma.
\qed

\begin{lemma}~\label{lem:lcimp->sp}
Let $H$ be a \mcsg of $\a$, then
$\a(H) \leq \opt$ and if $\a(H) = \opt$ then
$\a$ is strong.
\end{lemma}

\paragraph{Intuition for the proof Lemma~\ref{lem:lcimp->sp}.}
The proof of the Lemma~\ref{lem:lcimp->sp} is by induction on $s$,
the maximum index such that $x_s^\a > 0$.
The basis of the induction is straightforward.
The induction assumption states that for an \dist
$\be$ with~$s-1$ distinct strictly positive weights,
and \mcsg $H'$ of $\be$ we have $\be(H')\le \opt$ and
if $\be(H')= \opt$ then $\be$ is strong.

The main idea in the induction step is to show that
one can shift around some of the weight of $\a$ in order to get a new
\dist $\be$, such that $\be$ has exactly $s-1$ strictly positive distinct weights
and $\be(H') \ge \a(H)$ (or $\be(H') > \a(H)$), where $H'$ is
a \mcsg of $\be$.
Now since $\be$ has $s-1$ strictly positive distinct weights 
the induction assumption applies to it and hence $\be(H) \leq opt$.
This in turn implies that $\a(H) \le opt$.
Now by the above if $\a(H) = opt$ then also $\be(H)=opt$ and
hence by the induction assumption $\be$ is strong.
With a bit of extra work this leads to $\a$ being strong.

The induction step consists of two separate cases.
In the first case it is assumed that $cr_s^\a \leq \CR(\cup_{i=1}^{s-1}E_i^\a)$, 
in the second $cr_s^\a > \CR(\cup_{i=1}^{s-1}E_i^\a)$.

In the first case, by taking all the weight assigned by $\a$ to the edges of
 $E_s^\a$ and distributing it equally
among the edges in $\cup_{i=1}^{s-1}E_i^\a$, one gets a new \dist
 $\ga$ that has~$s-1$ distinct strictly positive weights
and  $\a(H) \leq \ga(H')$, where $H'$ is a \mcsg of $\ga$.
In the second case, one obtains the new \dist from $\a$
in the following way.
A constant amount of weight $\chi$ is reduced from each edge in 
$\cup_{i=1}^{s-1}E_i^\a$
and divide the total removed weight $\chi 
|\cup_{i=1}^{s-1}E_i^\a|$ equally
among the edges of $E_s^\a$ thus getting a new \dist $\delta$
where  $\a(H) < \delta(H'')$, where $H''$ is a \mcsg of $\delta$.
The value of $\chi$ is chosen so that $\delta$ gives the same 
weight to all the edges in $E_s^\a\cup E_{s-1}^\a$.
Therefore the number of strictly positive weights of $\delta$
is $s-1$.


\paragraph{\bf Proof of Lemma~\ref{lem:lcimp->sp}.}
If $s=1$ then by Proposition~\ref{prop:mcsg} we have
$\a(H)  = x_1^\a|E_1^\a| cr_1^\a 
            = cr_1^\a = \CR(E_1^\a) \leq \opt$.
Note that if equality holds then $cr_1^\a = \opt$ 
and hence $\a$ is strong.

Let $s>1$.
The induction assumption is that for every \dist $\be$ 
that has~$s-1$ strictly positive weights, we have 
$\be(H') \leq \opt$ for $H'$ that is a \mcsg of 
$\be$ and if $\be(H') = \opt$ then $\be$ is strong.

For the inductive step we deal with two cases separately.
In case {\bf (a)} we assume that 
\begin{equation}
\label{eq:strong-a}
cr_s^\a \leq \frac{\sum_{i=1}^{s-1}cr_i^\a
|E_i^\a|}{\sum_{i=1}^{s-1}|E_i^\a|}\ .
\end{equation}
In case {\bf (b)} we assume that 
\begin{equation}
\label{eq:strong-b}
cr_s^\a > \frac{\sum_{i=1}^{s-1}cr_i^\a
|E_i^\a|}{\sum_{i=1}^{s-1}|E_i^\a|}\ .
\end{equation}
Note that we chose to write the more cumbersome $\frac{\sum_{i=1}^{s-1}cr_i^\a 
|E_i^\a|}{\sum_{i=1}^{s-1}|E_i^\a|}$ instead of 
$\CR(\cup_{i=1}^{s-1}|E_i^\a|)$ as this form serves our
purpose better.
\begin{paragraph}
{\bf (a)} 
Set 
$$
\rho =
\frac{x_s^\a|E_s^\a|}{\sum_{i=1}^{s-1}|E_i^\a|}\
,
$$
which is the total weight of $E_s^\a$ divided equally
among all edges in $\cup_{i=1}^{s-1}E_i^\a$.
Define $\ga:E(G)\rightarrow \reals$ so that $\ga(e)= \a(e)+\rho$
for every $e\in \cup_{j=1}^{s-1}E_j^\a$ and $\ga(e)=0$ for every 
$e \in E(G) \setminus \cup_{j=1}^{s-1}E_j^\a$.
According to this definition
$$
\sum_{e\in E} \ga(e) = 
\sum_{e\in E} \a(e) - \sum_{e\in E_s^\a}x_s^\a +
\sum_{e\in \cup_{i=1}^{s-1}E_i^\a}\rho\ .
$$
Since $\a$ is an \dist we can replace 
$\sum_{e\in E(G)} \a(e)$ with $1$.
Doing so in the above equation in addition to replacing $\rho$ by
its value gives us
$$\sum_{e\in E(G)} \ga(e) =
1 - \left(x_s^\a|E_s^\a| -\frac{x_s^\a|E_s^\a|}{\sum_{i=1}^{s-1}|E_i^\a|}\sum_{i=1}^{s-1}|E_i^\a|
\right) = 1.
$$
By definition $\ga$  has 
exactly $s-1$ strictly positive weights and hence, $\ga$ is an \dist
and the induction assumption applies to $\ga$.
Let $H'$ be a \mcsg\ of $\ga$. 
By the induction assumption $\ga(H') \leq \opt$. 
We next show that $\a(H) \leq \ga(H')$ and hence $\a(H) \leq \opt$.
According to the construction of $\ga$ we have 
$x_i^\ga > x_j^\ga$ if and only if $x_i^\a > x_j^\a$ 
for any $i,j \in \{1,\dots,s-1\}$ and therefore
$E_i^\a = E_i^\ga$ for $i= 1,\dots,s-1$.
 This in turn implies that 
$cr_i^\a = cr_i^\ga$ for $i= 1,\dots,s-1$.
According to Fact~\ref{fact:mcsg} we have
$$\ga(H') = \sum_{i=1}^{s-1}x_i^\ga cr_i^\ga |E_i^\ga|\ .$$
By replacing $E_i^\a$ with  $E_i^\ga$ and $cr_i^\a$ with $cr_i^\ga$
and $x_i^\ga$ with $x_i^\a+\rho$ for $i=1,\dots,s-1$
 we get 
$$ \ga(H') = \sum_{i=1}^{s-1}(x_i^\a + \rho) cr_i^\a |E_i^\a|.$$
This implies 
$$ \ga(H') = \sum_{i=1}^{s}x_i^\a cr_i^\a |E_i^\a| - x_s^\a cr_s^\a |E_s^\a|
+ \rho \sum_{i=1}^{s-1}cr_i^\a |E_i^\a|.
$$
By Fact~\ref{fact:mcsg}, we can replace
$\sum_{i=1}^{s}x_i^\a cr_i^\a |E_i^\a|$ by $\a(H)$. 
By also replacing
$\rho$ with its value we have
\begin{equation}~\label{equ:lc2}
 \ga(H') = \a(H) - x_s^\a|E_s^\a|\left( cr_s^\a
- \frac{\sum_{i=1}^{s-1}cr_i^\a
|E_i^\a|}{\sum_{i=1}^{s-1}|E_i^\a|}\right).
\end{equation}
Now by \eqref{eq:strong-a} and \eqref{equ:lc2} we have $\a(H) \leq \ga(H')$.

Assume that $\a(H) = \opt$. 
Since $\a(H) \leq \ga(H') \leq \opt$  we get $\ga(H') = \opt$. 
Thus by the induction assumption $\ga$ is strong and hence
$cr_i^\ga = \opt$ for $i=1,\dots,s-1$.
Since $cr_i^\a = cr_i^\ga = \opt$ for $i=1,\dots,s-1$, 
to conclude that $\a$ is strong. 
Thus, we only need to show that $cr_s^\a = \opt$.
By replacing $\a(H),\ga(H'),cr_1^\a,\dots cr_{s-1}^\a$ with $\opt$ 
in \eqref{equ:lc2} we get
$$
cr_s^\a = \frac{\sum_{i=1}^{s-1}\opt
|E_i^\a|}{\sum_{i=1}^{s-1}|E_i^\a|} = \opt.
$$
\end{paragraph}

\begin{paragraph}
{\bf (b)}
Let 
$$
\chi = (x_{s-1}^\a - x_s^\a) (1 +
\frac{|E_s^\a|}{\sum_{i=1}^{s-1}|E_i^\a|})^{-1}\ .
$$

Let $\delta$ be such that 
\begin{equation}
\delta(e)=
\begin{cases}
\a(e)+\chi & \text{ if } e\in E_s^\a\ , \\ 
\a(e)-\chi\frac{|E_s^\a|}{\sum_{i=1}^{s-1}|E_i^\a|} & \text
{ if } e\in \cup_{i=1}^{s-1} E_i^\a \ , \\
0 & \text{otherwise.}
\end{cases}
\end{equation}

Note that $\chi$ is such that 
$ x_s^\a+\chi = x_{s-1}^\a-\chi\frac{|E_s^\a|}{\sum_{i=1}^{s-1}|E_i^\a|}$.
Consequently, $\delta$ assigns the same weight to each edge in 
$E_{s-1}^\a \cup E_{s}^\a$ and hence $\delta$ has
exactly $s-1$ strictly positive weights.

We next show that $\delta$ is an \dist.
By definition
$$\sum_{e\in E(G)} \delta(e) =
\sum_{e\in E(G)} \a(e) + \sum_{e\in E_s^\a}\chi - \sum_{e\in \cup_{j=1}^{s-1}}\chi\frac{|E_s^\a|}{\sum_{i=1}^{s-1}|E_i^\a|}.
$$
Since $\a$ is an \dist we can replace $\sum_{e\in E(G)} \a(e)$ by  $1$
in the above to conclude
$$\sum_{e\in E(G)} \delta(e) =
1 + \chi \left( |E_s^\a| -
\frac{|E_s^\a|}{\sum_{i=1}^{s-1}|E_i^\a|}\sum_{i=1}^{s-1}|E_i^\a|\right) = 1.
$$
Note that by the choice of $\chi$ we have $\delta(e) > x_s^\a$
for every $e \in \cup_{i=1}^{s} E_i^\a$ and since $\delta(e) = 0$ for any other edge
all the weights $\delta$ assigns are non-negative.
Thus $\delta$ is an \dist.
Let $H'$ be a \mcsg\ of $\delta$.
Now as $\delta$ is an \dist with $s-1$ strictly
positive weights, by the induction assumption we have
$\delta(H') \leq \opt$.
We conclude the claim by showing that $\a(H)<\delta(H')$.

By Fact~\ref{fact:mcsg} we have
\begin{equation}~\label{equ:lc5}
\delta(H') = \sum_{i=1}^{s-1}x_i^\delta cr_i^\delta |E_i^\delta|.
\end{equation}
According to the construction of $\delta$ we have 
$x_i^\delta > x_j^\delta$ if and only if $x_i^\a > x_j^\a$ 
for any $i,j \in \{1,\dots,s-2\}$ and therefore
$E_i^\a = E_i^\delta$ for $i=1,\dots,s-2$ which in turn implies that 
$cr_i^\a = cr_i^\delta$ for $i=1,\dots,s-2$.
Thus $E_i^\a = E_i^\delta $ and $cr_i^\a = cr_i^\delta$ for $i= 1,\dots,s-2$.
Consequently by replacing $|E_{s-1}^\delta|$ with $|E_{s-1}^\a| + |E_{s}^\a|$
and $cr_i^\ga$ with $cr_i^\a$  
for $i=1,\dots,s-2$ in \eqref{equ:lc5} we get
\begin{equation}~\label{equ:lc6}
\delta(H') = 
\sum_{i=1}^{s-2}x_i^\delta cr_i^\a |E_i^\a| +
x_{s-1}^\delta cr_{s-1}^\delta (|E_{s-1}^\a| + |E_{s}^\a|).
\end{equation}
By definition of \crate we have 
$$
cr_{s-1}^\delta = \frac{|E_{s-1}^\a|
cr_{s-1}^\a + |E_{s-1}^\a|
cr_{s}^\a}{|E_{s-1}^\a|+|E_s^\a|}\ ,
$$
Plugging this in \eqref{equ:lc6} gives us
$$
\delta(H') = \sum_{i=1}^{s-2}x_i^\delta cr_i^\a |E_i^\a|
+ x_{s-1}^\delta cr_{s-1}^\a|E_{s-1}^\a|+
x_{s-1}^\delta cr_{s}^\a|E_{s}^\a|.
$$
Since $x_s^\delta = x_s^\a + \chi$ and $x_i^\delta = x_i^\a - \chi\frac{|E_s^\a|}{\sum_{i=1}^{s-1}|E_i^\a|}$ 
for $i=1,\dots s-2$ we get
\begin{equation}~\label{equ:lc7}
\delta(H') = \sum_{i=1}^{s}x_i^\a cr_i^\a |E_i^\a|
+ \chi |E_s^\a|\left(cr_s^\a - \frac{\sum_{i=1}^{s-1}cr_i^\a |E_i^\a|}{\sum_{i=1}^{s-1}|E_i^\a|}\right).
\end{equation}
By Fact~\ref{fact:mcsg} we can also replace $\sum_{i=1}^{s}x_i^\a cr_i^\a |E_i^\a|$
with $\a(H)$ in \eqref{equ:lc7}. 
This together with \eqref{eq:strong-b} implies that $\delta(H') > \a(H)$.
Note that as $opt > \delta(H')$ it can not be the case that $\a(H)=opt$.
\end{paragraph}
\qed



\section{Appendix: Prime partition}
\label{sec:prime-partition}

In this section we prove Theorem~\ref{thm:unique-prime-partition}.
To do so, we introduce a polynomial time algorithm 
that on input graph $G=(V,E)$ returns a partition of $E$, which 
afterwards we show is the prime-partition of $G$.
The algorithm uses oracle access to a routine
$\mathtt{PrimeSet}$ that 
given a graph returns its cut-rate and one of its prime sets.
This routine runs in time polynomial in the size of $G$ and is 
introduced in Subsection~\ref{subsec:PrimeSet}

\subsection{Construction of the prime-partition}

  \begin{algorithm}[H]
  \caption{Prime-partition construction}
  \label{alg-gd}
  \textbf{Input:} Graph $G$ .\\
  \textbf{Output:} Prime partition $\pp$.

  \begin{algorithmic}[1]
  \STATE $\pp \leftarrow \emptyset$
  \IF{$E(G) = \emptyset$}
  \RETURN $\pp$
  \ENDIF
  \STATE $i \leftarrow 1$
  \STATE $(\opt,P_i) \leftarrow\mathtt{PrimeSet}(G)$
  \STATE $G_i \leftarrow G$
  \REPEAT
  \STATE $\pp \leftarrow \pp \cup \{P_i\}$
  \STATE $G_{i+1} \leftarrow G_i\setminus P_i$
  \STATE $i \leftarrow i+1$
  \STATE $(c,P_i) \leftarrow\mathtt{PrimeSet}(G_i)$
  \UNTIL {$c< \opt$}
  \IF{$E(G_i)\neq \emptyset$}
  \STATE $\pp \leftarrow \pp \cup \{P_i\}$
  \ENDIF
  \RETURN $\pp$
 \end{algorithmic}
\end{algorithm}

Each computation done by Algorithm~\ref{alg-gd} requires
running time polynomial in $|V|$ including the calls to
$\mathtt{PrimeSet}$
according to Appendix~\ref{subsec:PrimeSet}.
Therefore, the only reason the running time of Algorithm~\ref{alg-gd} 
may be too long
is the repeat loop.
Note that, if $\mathtt{PrimeSet}$ returns an empty set, 
then it also sets $c=0$.
Thus, after any iteration of the repeat loop that is not the last
the number of edges of $G'$ is decreased by at least $1$.
Observe that $\mathtt{PrimeSet}$ returns $(0,\emptyset)$ and does not reach
the repeat loop if $\opt = 0$.
Thus only if $\opt > 0$ then the repeat loop is reached
and then the above ensures that it goes through at most $|E|$ iterations.
Hence the running time of Algorithm~\ref{alg-gd} is
polynomial in the size of $G$.

From here on in this section  $t = |\pp|$, where $\pp$ is the output
of Algorithm~\ref{alg-gd} on input graph $G$, and the elements
of $\pp$ are named as they were named by Algorithm~\ref{alg-gd},
thus $\pp=\{P_1,\dots,P_{t}\}$. 
In addition let $E^0 =\emptyset$ and for each $k=1,\dots,t$ let
$E^k = \cup_{i=1}^{k} P_i$ and
$r_k$ be the cut-rate of $P_k$ in $G\setminus E^{k-1}$.

\subsection{The output of Algorithm~\ref{alg-gd} is the prime partition}


\begin{proposition}~\label{prop:pp-strong}
There exists a \mdist $\be$ such that $\E(\be) = \pp$.
\end{proposition}
\begin{proof}
Set $\rho = \frac{1}{\sum_{i=1}^{t}(t-i)|E_i|}$ 
and let $\be:E\rightarrow \reals$, where
for each $i=1,\dots,t$ and $e \in P_i$ we have $\be(e) = (t-i)\rho$.
Observe that 
$$\sum_{e\in E}\be(e) = \sum_{i=1}^{t} x_i^\be|E_i^\be| =
\sum_{i=1}^{t} \rho (t-i)|E_i^\be| = \rho \sum_{i=1}^{t} (t-i)|E_i^\be|
= \rho \rho^{-1} = 1$$ and hence $\be$ is an \dist.
By definition $E_i^\be = P_i$ for $i=1,\dots,t$.
We next show that $r_i = \opt$ for $i=1,\dots,t-1$.
By Theorem~\ref{thm:strong-dist} this implies that $\be $ is a \mdist.

Algorithm~\ref{alg-gd} selects $P_1$ so that $r_1= \opt$.
Let $k< t$ and assume that $r_j = \opt$  for every $j < k$.
Hence by Fact~\ref{fact:composition} we have
$\CR(E^{k-1})=\opt$.
Consequently $r_k\leq \opt$ since otherwise by Fact~\ref{fact:composition} 
we get that $\CR(E^{k-1})>\opt$.
As $P_k$ is not the last set added to $\pp$ by Algorithm~\ref{alg-gd},
we have $r_k \geq \opt$ and hence it is the case that $r_k = \opt$.
\qed
\end{proof}

We next show that $\pp$ refines $\E(\a)$ for every \mdist $\a$.
We start with a simple case that we use later on to prove the general
result.
\begin{proposition}~\label{prop:refines-one}
If $\CR(E') = \opt$ for $E'\subseteq E$ then $\pp$ refines 
$\{E',E\setminus E'\}$. 
\end{proposition}
\begin{proof}
If $E' = E$ then the proposition holds trivially.
Hence we only need to prove the proposition holds when $E'\subset E$.
We show that 
$P_i\cap E' = \emptyset$ or $P_i\subseteq E'$ for every $i=1,\dots,t$.
Let $k\in \{1,\dots,t\}$.
If $E'\setminus E^{k-1} = \emptyset$ then $P_k\cap E' = \emptyset$.
Therefore we only need to deal with the case that 
$E'\setminus E^{k-1} \neq \emptyset$.
Assume this is indeed so.
By Proposition~\ref{prop:pp-strong}, we have $r_i =\opt$ for
 $i=1,\dots,t-1$ hence by Fact~\ref{fact:composition} we get 
$\CR(E^{k-1})=\opt$.
Since also $\CR(E')=\opt$, by
Proposition~\ref{prop:prime-set}, we have $\opt_{G\setminus E^{k-1}} = \opt$ 
and $\CR_{G\setminus E^{k-1}}(E'\setminus E^{k-1})=\opt$.
We separate the proof into two cases the first $k=1,\dots,t-1$ and
 in the second $k=t$.

Recall that Algorithm~\ref{alg-gd} selects
$P_k$ so that it is a prime-set in $G\setminus E^{k-1}$.
So now $\opt_{G\setminus E^{k-1}} = \opt$ and
$\CR_{G\setminus E^{k-1}}(P_k) = \CR_{G\setminus E^{k-1}}(E'\setminus E^{k-1}) = \opt$.
Thus by Proposition~\ref{prop:prime-set} either 
$P_k\cap (E'\setminus E^{k-1}) = \emptyset$ or 
$P_k\subseteq (E'\setminus E^{k-1})$.
If $P_k\subseteq (E'\setminus E^{k-1})$ then $P_k\subseteq E'$,
and if $P_k\cap (E'\setminus E^{k-1}) = \emptyset$ then 
$P_k\cap E^{k-1} = \emptyset$ because $P_k\cap E^{k-1} = \emptyset$.

Assume $k=t$ and for the sake of contradiction that
 $E'\setminus E^{t-1}\ne \emptyset$.
Since we have shown that $\opt_{G\setminus E^{t-1}} = \opt$
and $\CR_{G\setminus E^{t-1}}(E'\setminus E^{t-1})=\opt$ 
it is the case that
$G\setminus E^{t-1}$ has a prime-set 
that has cut-rate $\opt$ in $G\setminus E^{t-1}$.
This subset is strictly contained in $E\setminus E^{t-1}$ Proposition~\ref{prop:prime-set} implies that every prime set in $E\setminus E^{t-1}$ is
strictly contained in $E\setminus E^{t-1}$.
Hence, Algorithm~\ref{alg-gd} would have found a prime-set $E^*\subset E\setminus E^{t-1}$ and added it to $\pp$. That is $E^*\in \pp$.
Yet this can not be since by construction $\pp$ is a partition of~$E$.
\qed
\end{proof}

\begin{proposition}~\label{prop:prime-refines}
If $\ga$ is a \mdist then $\pp$ refines $\E(\ga)$.
\end{proposition}
\begin{proof}
Let $t = |\E(\ga)|$.
Recall that since $\ga$ is a \mdist by definition
for $i=1,\dots,t-1$ we have $cr_i^\ga = \opt$.
If $t=1$ then the only set in $\E(\ga)$ is $E$ and 
hence the lemma trivially holds.
By Proposition~\ref{prop:refines-one} the lemma also holds when $|\E(\ga)|=2$.
Assume by way of induction that proposition holds for any partition
$\es = \{E_1,\dots,E_{t-1}\}$ of $E$ such that $\CR_{G\setminus\cup_{i=1}^\ell E_i}(E_\ell)=\opt$.
Let $\es_1 = \{\cup_{i=1}^{t-1} E_i^\a, E_t^\a\}$ and
$\es_2 = \{E_1,\dots,E_{t-2}^\a, E_{t-1}^\a\cup E_t^\a\}$.
Note that if $\pp$ refines both $\es_1$ and $\es_2$ then it refines~$\es$.
By the induction assumption, $\pp$ refines $\es_2$.
By Fact~\ref{fact:composition} we have that $\CR(\cup_{i=1}^{m-1}E_i)=\opt$
and therefore $\pp$ refines $\es_1$ by Proposition~\ref{prop:refines-one}.
\qed
\end{proof}

\subsection{$\mathtt{PrimeSet}$}~\label{subsec:PrimeSet}

\eat{
\todo{prime-set is not a good name we need something else 
in order not to confuse with prime partition}
}
In this section we explain the subroutine for finding a
minimal optimal set, which we call a $\mathtt{PrimeSet}$.
We assume that the graph is connected, in case it is not
connected we run the routine separately on each connected
component and return the $\mathtt{PrimeSet}$ (and value
$\opt$ that achieves the largest $\opt$ among these 
connected components.
If there is more than one, pick one arbitrarily.
By Fact~\ref{fact:composition}, 
 the cut-rate of $\cup_{i=1}^{m-1}E_i$ in $G$ is $\opt$
and therefore by Proposition~\ref{prop:refines-one}, we have
that $\pp$ refines $\es_1$.
\eat{
\todo{Take into account the case that the graph is not 
connected or does not have any edges.}
}
For our goal, we extend the notion of cut-rate of a graph 
to edge weighted graphs.

\begin{definition}~\label{def:weighted-set-cut-rate}
Let $E'\subseteq E$ and $\omega:E\rightarrow \reals^+$.
The {\em cut-rate} of $E'$ in $G,\omega$ is denoted by $\CR_\omega(E')$ and defined
as follows.
\begin{equation}
\CR_\omega(E') \coloneqq
\begin{cases}
 \frac{C_G(E')-\cge}{\omega(E')}\quad &\text{if
} |V|>1 \text{ and } |E'| > 0\ ,\\
0 \quad &\text{otherwise}\ .
\end{cases}
\end{equation}
The \crate of $G,\omega$ where $\omega:E\rightarrow \Re^+$ is defined as
\begin{equation}
\opt_\omega \coloneqq \max_{E'\subseteq E} \CR_\omega(E')\ 
\end{equation}
\end{definition}

There exists strongly polynomial algorithms in~\cite{Cun85,Tru91,CC94} that on
$G,\omega$ returns $\opt_\omega$. We shall assume from here on that
$\opt_\omega$ is given and omit the fact that this is done by the
mentioned algorithm.

A prime-set of $G$ is found as follows.
If $E=\emptyset$ then stop and return $(0,\emptyset)$.
Otherwise, set $\omega:E\rightarrow \Re^+$ so that $\omega(e) = 1$ 
for every $e\in E$.
Note that in this case $\opt = \opt_\omega$ and hence we assume 
$\opt$ is known.
Set $\omega' = \omega$.
Next iterate $e$ over the elements of $E$ according to some arbitrary order 
and in each iteration do the following.
Set $\omega'(e)$ to be $2$ and if $\opt_{\omega'} = \opt$ then
set $\omega$ to be $\omega'$ and otherwise set $\omega'$ to be $\omega$.
That is, $\omega(e)$ is changed only if $\opt_{\omega'} = \opt$
and otherwise remains the same.
After the iterative process is over set $E' = \{e\in E\mid \omega(e)=1\}$ and return $(\opt,E')$.
Note that the total number of operations done is polynomial in the
size of $G$ and so is the running time.
To show that indeed this achieves our goal, we only need to
prove that $E'$ is a prime set of~$G$.

Let us look at any fixed iteration over $e$.
By definition $\omega$ is changed only if $\opt_{\omega'} = \opt$
and then it is set to $\omega'$.
This implies that there exists $E^*\subseteq E$ such that
$cr_{\omega'}(E^*) = \opt$. 
Now it can not be the case that $\omega'(e) =2$ for some
$e\in E^*$, since this would imply that $cr(E^*) > \opt$.
Consequently, $cr_{\omega}(E^*) = \opt$.
This is true for any fixed iteration and hence also for the last.
Therefore, there exists $E'' \subseteq E'$ such that
$cr(E'') = \opt$.
Finally assume for the sake of contradiction that $E'' \subset E'$.
Let $e'\in E'\setminus E''$.
This implies that in the iteration dedicated to $e'$ we had
$\opt_{\omega'} \neq \opt$.
Since at this  stage $\omega'(e) = 1$ for every $E''$ it must be that
$\opt_{\omega'} > \opt$. Yet this can not be since the weights assigned
to each edge by $\omega'$ is at least as that assigned by $\omega$
and hence at every iteration $\opt_{\omega'} \leq \opt$.

\section{Appendix: Proof of Lemma~\ref{lem:deg-set}}

\label{sec:deg-set}

If $\CR(E) = \opt$, then an \dist that assigns equal weight to
all edges is a \mdist and so there is no degenerate set.

We now prove that if $\CR(E) \ne \opt$, then the degenerate set exists,
it is unique, and can be found in running time polynomial in the 
size of $G$.

Assume that $\CR(E) \ne \opt$.
By the definition of $\opt$, this can only happen if $\CR(E) < \opt$.
Let $\be$ be a \mdist  such that $\E(\be) = \pp$ and set
$t= |\pp|$.
By definition, $E_t^\be$ assigns strictly positive 
weights to the edges in each $E_i^\be$ for every $i=1,\dots,t-1$.
Hence, the only candidate for being the degenerate set is $E_t^\be$.
We next show that this is indeed the case.

Assume for the sake of contradiction that there exists a \mdist~$\ga$
that assigns strictly positive weights to the edges in $E_t^\be$.
Let $d = \min_{i\in\{1,\dots,t-1\}}\{x_i^\be - x_{i-1}^\be\}/10$ and
set $\delta = (1-d)\be + d\ga$ (the choice of 10 is
arbitrary).
Observe that $\delta$ has the same number of distinct weights as $\ga$, 
and $\E(\delta) = \E(\be)$ and we have $E_i^\delta = E_i^\be$ for
$i = 1,\dots,t$.
Let $H$ be a \mcsg of $\delta$. 
Since $\delta$ is a convex combination of \mdists it
is a \mdist and therefore, by
Corollary~\ref{cor:value-of-game}, we have
$\delta(H) = \opt$.
We next get the required contradiction by showing that $\delta(H) < \opt$.

Since $\E(\delta)$ is a partition of $E$ we have
\begin{equation}~\label{equ:deg1}
\delta(H) = \sum_{i=1}^t x_i^\delta|E(H)\cap E_i^\delta|.
\end{equation}

As $\delta$  is a \mdist, by Theorem~\ref{thm:strong-dist}, we have
$cr_i^\be = \opt$ for $i=1,\dots,t-1$ and hence 
 $cr_t^\be < \opt$, since otherwise,
 by Fact~\ref{fact:composition}, we have $\CR(E) \geq \opt$.
By Proposition~\ref{prop:mcsg}, we have 
$|E(H)\cap E_i^\delta| = \opt|E_i^\delta|$ for $i = 1,\dots,t-1$.
Applying this to~\eqref{equ:deg1} we get
\begin{equation}~\label{equ:deg2}
\delta(H) = x_t^\delta |E(H)\cap E_t^\delta|+
\opt\sum_{i=1}^{t-1} x_i^\delta|E_i^\delta|\ .
\end{equation}
Now, since $H$ is a \mcsg, $|E(H)\cap E_t^\delta|$ is the
minimum possible, 
which in this case is $cr_t^\delta|E_t^\delta|$.
Since $cr_t^\delta <\opt$, we get that $cr_t^\delta|E_t^\delta|<
\opt |E_t^\delta|$.
Thus, by replacing  $|E(H)\cap E_t^\delta|$ by $\opt |E_t^\delta|$
in \eqref{equ:deg2}, we get 
$$
\delta(H) < \opt\sum_{i=1}^t x_i^\delta|E_i^\delta| = \opt\ ,
$$
where the equality is because $\delta$ is an \dist.

We now explain how to compute $D$.
Once $\opt$ is known, one only needs to check if $\frac{|V|-1}{|E|} = \opt$.
If the answer is yes, then there is no degenerate set; if the answer is
no, then $D$ is the last set inserted to $\pp$ by Algorithm~\ref{alg-gd}.

\section{Appendix: Proof of Proposition~\ref{prop:omni-opt}}

\label{sec:omni-opt}

Let $H$ be an \ocsg and $\a$
such that $\pp$ refines $\E(\a)$ and $\a(e) = 0$ for every $e\in D$.
Thus for each $P\in \pp$ there exists $y^P$ such that $\a(e) = y^P$
for every $e\in P$.
Since $\pp$ is a partition of $E$, we have
$\a(H) = \sum_{P\in \pp} y^P|H\cap P|$ and
as $H$ is an \ocsg also $|H\cap P| = |P|\opt$ for every $P\in \pp$.
Consequently,
$$\a(H) = \sum_{P\in \pp} y^P|P|\opt =
\opt\sum_{P\in \pp} y^P|P|.$$
Now, as $\a$ is an \dist and $\sum_{P\in \pp} y^P|P|=1$, with
the above, we get $\a(H)=\opt$.
Now, if $\a$ is \mdist, by Corollary~\ref{cor:value-of-game}, the
value of the game is $\opt$ and  $H$ is a \mcsg of $\a$. 

\qed

\section{Appendix: Partial Order}

\label{sec:partial-order}

\subsection{Proof of Theorem~\ref{thm:prime-order}}

\label{sec:prime-order}

Note that if $\pp = 1$ then the theorem trivially holds
hence we assume $\pp > 1$.

Let $\a$ be a \mdist.
By Lemma~\ref{lem:deg-set} we have $\a(e) = 0$ for every $e\in D$.
Assume for the sake of contradiction that $\a$ does not agree with $\pr$.
By Definition~\ref{def:prime-partition}, we have that $\pp$ refines 
$\a$ and hence since $\a$ does not agree with $\pr$ there exist 
$P\in \pp\setminus \{D\}$ that leads to $P'\in \pp\setminus \{D\}$, 
an \ocsg $H$, $e\in E(H)\setminus P$ and $e'\in P'\cap E(H)$
such that $\a(e) < \a(e')$ and
$H' = (H\setminus \{e'\})\cup \{e\}$ is a \csg.
Since $H$ is an \ocsg and $\a$ a \mdist by
Proposition~\ref{prop:omni-opt} 
we have $\a(H) = \opt$.
Thus $\a(H') = \a(H) + (\a(e)-\a(e')) < \opt$.
This is in contradiction to Fact~\ref{cor:value-of-game}, 
which implies
that $\a(H')\geq \opt$ since $\a$ is a \mdist.

Assume $\a$ is an \dist that agrees with $\pr$.
We next show that this implies that $\a$ is a \mdist.

Let $m$ be the number of the strictly positive weights of $\a$.
Assume by way of contradiction that $\a$ is not a \mdist.
By Theorem~\ref{thm:strong-dist} this can only happen if
there exists  $i\in \{1,\dots,m\}$ such that $cr_i^\a\neq \opt$.
Let $\ell$ be the smallest element in $\{1,\dots,m\}$ such that 
$cr_\ell^\a\neq \opt$.
Let $E' = \cup_{i=1}^\ell E_i^\a$.
We show next that $\CR(E')<\opt$.
If $\ell =1$ then $cr_\ell^\a \leq \opt$ since $cr_\ell^\a$ is the
cut-rate of $E_\ell$ in $G$. 
Thus in this case $E' = E_\ell$, and the goal is achieved. 
Assume that $\ell>1$.
By the minimality of $\ell$, we have that $cr_i^\a = \opt$ for
$i=1,\dots,\ell-1$.
If $cr_\ell^\a > \opt$ then by Fact~\ref{fact:composition} we have
$\CR(\cup_{i=1}^\ell E_i^\a)>\opt$ which is a contradiction
to the definition of $\opt$.
Thus, as $cr_\ell^\a\neq \opt$ we have $cr_\ell^\a < \opt$ and hence again
by Fact~\ref{fact:composition} the $\CR(E')<\opt$.

Let $C_1,\dots,C_s$ be the connected components of $G\setminus E'$.
Let $H$ be an \ocsg and let $H_1,\dots,H_r$ be the connected components
of $E(H)\setminus E'$.
Note that for each $i\in \{1,\dots,r\}$ there exists a unique $j\in\{1,\dots,s\}$
such that $E(H_i)\subseteq E(C_j)$.
For each $j\in\{1,\dots,s\}$ set $I_j$ to be the set of all $i\in \{1,\dots,r\}$ such that
$E(H_i)\subseteq E(C_j)$.
Assume that $s<r$, we shall show afterwards that this is indeed true.
By the pigeon-hole principle there exists $j\in \{1,\dots,r\}$ such that 
$|I_j|>1$.
Since $C_j$ is a connected component and $H$ a \csg of $G$ there
exist $x,y\in I_j$ and $e=\{u,v\}\in E(C_j)\setminus \cup_{i=1}^{|I_j|}E(H_i)$ such
that  $u\in V(H_x)$ and  $v\in V(H_y)$. 
Since $H$ is a \csg there is a path in $H$ between $u$ and $v$
this path contains edges not in $E(C_j)$ since $u,v$ are in different connected
components of $H\setminus E'$.
Thus this path contains an edge $e'\in E'$ since only edges from $E'$ connect the 
vertices of $C_j$ to the rest of the graph.
Consequently $(H\setminus \{e'\})\cup\{e\}$ is a \csg of $G$.
Let $P,P'\in \pp$ be such that $e\in P$ and  $e' \in P'$. 
By the above $P$ leads to $P'$. Yet, this can not be since 
$\a(e) < \a(e')$ and $\a$ agrees with $\pr$. 

It remains to show that indeed $s<r$.
By the definition of cut-rate the number of connected components $s$ in $G\setminus E'$
is $\CR(E')|E'|$, which is strictly less than $\opt|E'|$.
Now as $\a$ agrees with $\pr$ we know that $\pp$ refines $E'$.
Hence $E'$ is the union of sets in $\pp \setminus \{D\}$.
Consequently, by the definition of a \ocsg, we have $E(H)\cap E' = \opt|E'|$.
Hence $r=\opt|E'|$ because the number of connected components in $H\setminus E'$ 
is the number of edges in $E(H)\cap E'$.
\qed

\subsection{Proof of Theorem~\ref{thm:compute-po}}

\label{sec:compute-po}

\begin{definition}~\label{def:ancestors}
We say that $P\in \pp$ is an ancestor of $P'\in \pp$ if
there is a chain in $\po$ from $P$ to $P'$.
\end{definition}

We show that for each $P\in \pp\setminus \{D\}$ we can
find all of the ancestors of $P$. 
Once we know the ancestors of each element $\pp\setminus \{D\}$ 
finding the parent of each such element is easy.
An ancestor $P$ of $P'$ is also a parent of $P'$ if there does
not exist an $P^*$, that is neither $P$ nor $P'$, 
such that $P$ is an ancestor of $P^*$ and
$P^*$ is an ancestor of $P'$.
Checking this for each pair element and each one of its ancestors
requires running time that is polynomial in the size of $G$.

To achieve our goal we need the following proposition.
\begin{proposition}~\label{prop:order}
Let $\pp^*\subseteq \pp\setminus\{D\}$,  and $P^* \in \pp^*$ and 
set $E^* = \cup_{P\in\pp^*} P$ then
\begin{itemize}
\item
$cr(E^*) = opt$ if $\pp^*$ contains only $P^*$ and all its ancestors.
\item
If $cr(E^*) = opt$ and $\pp^*$ contains an element that is not $P^*$ 
or one of its ancestors then it also contains such a $P$ for which
$cr(E^*\setminus P) = \opt$.
\end{itemize}
\end{proposition}

\begin{proof}
Set $\a:E(G)\rightarrow \reals$ so that $\a(e) = \frac{1}{|E^*|}$ if 
$e\in E^*$ and $\a(e) = 0$ otherwise. 
Now $\a$ is an \dist, $\pp$ refines $\E(\a)$ and $E_1^\a = E^*$,
$E_2^\a = E\setminus E^*$ and $\a(e)=0$ for every $e\in D$.

We now prove the first item.
For any parent and its child
if the child is in $\pp^*$ it is either $P^*$ or one of its ancestors.
Thus the parent is also an ancestor of $P^*$ and hence is also in 
$\pp^*$.
Consequently $\a$ agrees with $\po$ and hence by
Theorem~\ref{thm:prime-order}, we have that
$\a$ is a \mdist.
This in turn by Theorem~\ref{thm:strong-dist} implies $cr(E^*) = \opt$.

We now prove the second item.
Assume $cr(E^*) = opt$ and $\pp^*$ contains an element that is not $P^*$ 
or one of its ancestors.
Then there exists $P\in \pp^*$ that is not an ancestor of any
other element in $\pp^*$.
Since $cr(E^*) = opt$ by  Theorem~\ref{thm:strong-dist} $\a$ is 
 a \mdist. 
Hence by Theorem~\ref{thm:strong-dist} $\a$ agrees with $\po$.

Set $\be:E(G)\rightarrow \reals$ so that $\be(e) = \frac{1}{|E^*|}$ if 
$e\in E^*$ and $\be(e) = 0$ otherwise. 
Now $\be$ is an \dist, $\pp$ refines $\E(\be)$ and $E_1^\be = E^*$,
$E_2^\be = E\setminus E^*$ and $\be(e)=0$ for every $e\in D$.
The only way that $\be$ does not agree with $\po$ is
 if a child of $P$ is in $\pp^*\setminus \{P\}$, yet
this can not be, since $P$ is not an ancestor of any element in $\pp^*$.
Thus, $\be$ agrees with $\po$ and hence, by
Theorem~\ref{thm:strong-dist}, $\be$
 is a \mdist.
By Theorem~\ref{thm:strong-dist}, 
this implies that $cr(E^*\setminus P) = \opt$.
\qed
\end{proof}

We next show how to find the ancestors of $P'$.
Set $\pp' = \pp\setminus \{D\}$ and $E' = \cup_{P\in\pp'} P$.
If there exists $P^*\in \pp'$ such that $P^* \ne P'$
and $\CR(E'\setminus P)=\opt$ remove it from $\pp'$ and
recompute $E'$.
Repeat until no such element is found.

Note that this requires $|V|$ repetitions each taking a polynomial time
in the size of~$G$. Consequently, the running time is polynomial 
in the size of~$G$.

When $\pp = \pp\setminus \{D\}$ we have $\CR(E') = \opt$ because
of the following.
By definition there exists a \mdist $\be$ such that $\E(\be) = \pp$.
Note that $\pp'$ is all the non-degenerate sets in $\E(\be)$
and hence $\CR(E') = \CR(\cup_{i=1}^{m}E_i^\be)$.
By Theorem~\ref{thm:strong-dist} $cr_i^\be =\opt$ for $i=1,\dots,m$,
where $m$ is the maximal index such that $x_m^\be > 0$.
Hence according to Fact~\ref{fact:composition} we have 
$\CR(\cup_{i=1}^{m}E_i^\be) = \opt$.

Finally we show that at the end what remains in $\pp'$ is only $P'$ and
all its ancestors.
The set $P'$ is never removed from $\pp'$.
By Proposition~\ref{prop:order} for any ancestor $P^*$ of $P'$ 
it is the case that $\CR(E'\setminus P^*) < \opt$ and hence
none of the ancestors of $P'$ are ever removed.
Also by Proposition~\ref{prop:order} as long as $\pp'$ does not contain
only $P'$ and each one of its ancestors there exists a $P^*$ such that
$\CR(E'\setminus P^*) = \opt$ and hence such an element will be removed.
Thus only $P'$ and each one of its ancestors are never removed and consequently
they are the only elements remaining in $\pp'$ at the end of the process.
\qed

\section{Appendix: Proof of Lemma~\ref{lem:prime-dist}}

\label{sec:prime-dist}

Let $\beta$ be a \mdist such that one of the following holds
\begin{enumerate}
\item
There exists $P\in \pp\setminus\{D\}$ that is a parent of 
$P'\in \pp\setminus\{D\}$ such that $\beta(e) = \beta(e')$ 
for every $e\in P$ and $e'\in P'$.
\item
There exist $P\in \pp\setminus\{D\}$ such that $\beta(e) = 0$ 
for every $e\in P$.
\end{enumerate}
We shall show that $\be$ has a \mcsg that is not an \ocsg.
Afterwards we shall show that for every $\ga$ for which both conditions
do not hold, every \mcsg of $\ga$ is an \ocsg.
According to Proposition~\ref{prop:omni-opt}, every \ocsg is a \mcsg of $\ga$,
this means that such $\ga$ are the only \mdists that have the minimum
possible number of \mcsgs.

Assume the first condition holds for $\be$.
By the Definition~\ref{def:parent} there exists an \ocsg $H$
 and edges $e_1 \in P\setminus E(H)$, $e_2 \in P'\cap H$ such that 
$H' = (T\cup \{e_1\})\setminus \{e_2\}$ is a \csg.
Observe that $H'$ is not an \ocsg of $\be$ but is a \mcsg of $\be$
since $\be(H') = \be(H) = opt$. 
Assume the second condition holds for $\be$. 
Let $H$ be an \ocsg.
Recall we assumed $\opt<1$ and hence as $H$ is an \ocsg we
have $|E(H)\cap P| = \opt |P| < |P|$ and therefore
there exists $e\in P\setminus E(H)$.
Since $H$ is a \mcsg of $\be$ by Proposition~\ref{prop:omni-opt} 
and  $\be(e) = 0$ we also have $H\cup\{e\}$ is a \mcsg of $\be$.
Note that $H\cup\{e\}$ is not an \ocsg.

Let $\ga$ be some \mdist for which the above two conditions do not hold.
That is, $\ga(e)>0$ for every $e\in E\setminus D$ and
$P\in \pp\setminus \{D\}$ that is a parent of $P'\in
\pp\setminus \{D\}$ and every $e\in P$, $e'\in P'$
we have $\ga(e) > \ga(e')$.

From here on let $H$ be a \mcsg of $\ga$.
We next show that $H$ is an \ocsg.
Let $m$ be the number of distinct strictly positive values of
$\ga$ and set
$\pp_i = \{P\in \pp \mid P\in E_i^\ga\}$ for $i=1,\dots,m$.
Note that by the definition of $\ga$ for every $P\in \pp\setminus\{D\}$
there exists $i\in \{1,\dots,m\}$ such that $P\in \pp_i$. 
Assume by way of contradiction that $H$ is not an \ocsg.
Let $k$ be the minimum integer such that there exists $P\in \pp_k$
for which $|H\cap P| \ne |P|\opt$.
Since $H$ is a \mcsg and $\ga$ a \mdist by Proposition~\ref{prop:mcsg} 
we have $|H\cap E_k^\ga| = \opt|E_k^\ga|$.
Since $\pp$ refines $\E(\ga)$ we also have
$|H\cap E_k^\ga| = \sum_{E'\in \pp_k}|H\cap E'|$ and 
$|E_k^\ga| = \sum_{E'\in \pp_k}|E'|$ and hence
$$\sum_{E'\in \pp_k}|H\cap E'| = \opt\sum_{E'\in \pp_k}|E'|$$ 
Therefore the fact that $|H\cap P| \ne |P|opt$ implies that there exists
$P'\in \pp_k$ such that  $|H\cap P'| < |P'|opt$. Let $P'$ be such a set.

Let $E^*$ be the union of $P'$ and all its ancestors (see Definition~\ref{def:ancestors} in Section~\ref{sec:partial-order}).
We next show that $\CR(E^*)=\opt$.
Set $\a:E(G)\rightarrow \reals$ so that $\a(e) = \frac{1}{|E^*|}$ if 
$e\in E^*$ and $\a(e) = 0$ otherwise. 
Observe that $\a$ is an \dist, $\pp$ refines $\E(\a)$ and $E_1^\a = E^*$,
$E_2^\a = E\setminus E^*$ and $\a(e)=0$ for every $e\in D$.
Now for any parent and its child
if the child is contained $E^*$ it is either $P'$ or one of its ancestors.
Thus the parent is also an ancestor of $P'$ and hence is also in 
$E^*$.
Consequently, $\a$ agrees with $\po$ and hence by
Corollary~\ref{cor:prime-order} we have that $\a$ is a \mdist.
This in turn by Theorem~\ref{thm:strong-dist} implies $\CR(E^*) = \opt$.

Note that because of the strict weight inequalities,
all the ancestors of $P'$ are elements in one of the sets 
$\pp_1,\dots,\pp_{k-1}$.
Thus for any ancestor $P^*$ of $P'$ we have 
$|H\cap P^*| = opt|P^*|$.
Consequently, $|H\cap E^*| < opt |E^*|$ yet this can not be
since $opt |E^*|$ is the minimum number of edges a \csg can have
in $E^*$.
\qed

\section{Appendix: The nucleolus}

\label{sec:Nucleolus}

\begin{proof}
Let $\kappa,\nu$ be as stated in the theorem and $H$ an \ocsg and
$t$ the number of layers.
Observe that 
$$\sum_{e\in E} \nu(e) = \sum_{i=1}^{t} i \cdot |L_i|\cdot \kappa =
\kappa\sum_{i=1}^{t} i \cdot |L_i| = \kappa \cdot \kappa^{-1} = 1$$
and hence $\nu$ is an \dist.
Note that by definition $\pp$ refines $\E(\nu)$ and 
for any $e\in P\in \pp$ that is a parent of $e'\in P'\in \pp$
we have $\nu(e) > \nu(e')$ and hence $\nu$ is a \mdist and specifically a
\pdist.

We now show that the weight of the second best strategy of the hider 
is $\opt + k$. 
Afterwards we show that the only $\nu$ is the only \pdist
for which the weight of the second best strategy of the hider 
is at least $\opt + k$.

Let $P\in \pp$ be such that $P\subseteq L_1$.
Since $opt<1$ Proposition~\ref{prop:mcsg} implies that there exists
$e \in P \setminus E(H)$.
By the definition of $\nu$ we have $\nu(e) = \kappa$.
By Proposition~\ref{prop:omni-opt} $\nu(H) = opt$ and hence 
 $\nu(H\cup\{e\}) = opt+\kappa$.
Note that since $\nu(e')$ is a  multiple of $\kappa$ for
every $e'\in E$ there does not exist a \csg $H'$ such that
$opt< \nu(H') < \opt+\kappa$.

Let $\a$ be a \pdist such that
the second smallest weight of a \csg is $opt+\kappa$.
We shall prove by induction on $\ell$ that $\a(e) > \ell\kappa$ 
for every $e\in L_\ell$ and every $\ell = 1,\dots,t$.
Since the only \pdist that satisfies this conditions is $\nu$
this implies that $\a = \nu$.

Assume for the sake of contradiction that there exists $e\in P \in L_1$ 
such that $\a(e) < \kappa$. 
Since $opt<1$ Proposition~\ref{prop:mcsg} implies that there exists
$e' \in P \setminus E(H)$.
By Proposition~\ref{prop:omni-opt} we have $\nu(H) = opt$ and hence
because $\a(e') = \a(e) < \kappa$ we get
 $\nu(H\cup\{e\}) = opt + \a(e') < opt+\kappa$.
In addition as $\a$ is a \pdist $\a(e')> 0$
and thus $\nu(H\cup\{e\}) > opt$.
Yet the assumption was that the hiders second best response to $\a$
is at least $\opt+\kappa$.

Assume by way of induction that for $\ell-1$ we have $\a(e) > (\ell-1)\kappa$ 
for every $e\in L_{\ell-1}$.
Assume for the sake of contradiction that there exists $e\in P\in L_\ell$ 
such that $\a(e) < \ell\cdot\kappa$. 
By the definition of $L_\ell$, there exists $P'\subseteq L_{\ell-1}$
such that $P$ is a parent of $P'$.
Consequently, there exists an \ocsg $H'$, $e'\in P\setminus E(H')$ and
$e''\in P'\cap E(H')$ such that $H^*= (H'\setminus \{e''\})\cup \{e'\}$ is
and spanning tree of $G$.
Observe that $\a(H^*) = \a(H') + \a(e')-\a(e'')$.
By Proposition~\ref{prop:omni-opt}, we have $\a(H') = opt$ and
hence $\a(H^*) = opt + \a(e')-\a(e'')$.
By the induction assumption, we have $\a(e'') \geq (\ell-1)\kappa$
and therefore as $\a(e'') = \a(e) < \ell\cdot\kappa$
we get $\a(H^*) < opt + \kappa$.
In addition as $\a$ is a \pdist we have $\a(e')> \a(e'')$
and therefore $\nu(H^*\cup\{e\}) > opt$.
Yet the assumption was that the hider's second best response to $\a$
is at least $\opt+\kappa$.
\qed
\end{proof}

\section{Appendix: Extreme points}

\label{s-extreme}


\begin{definition}~\label{def:minimal-closure}
We say $\mathcal{B} \subseteq \pp\setminus \{D\}$ is \emph{closed} if
for every $P\in \B$ all of the ancestors of $P$ are also in $\B$.
Set $A(\B) = \cup_{E'\in \B}E'$.
We say a closed set $\B$ is \emph{minimal} if there do not exist
 any closed sets
$\B_1,\B_2\subset \pp\setminus \{D\}$ such that
$A(\mathcal{B}) = A(\mathcal{B}_1)\cup A(\mathcal{B}_2)$.
\end{definition}

\begin{definition}~\label{def:extreme-imputation}
An \dist $\a$ is an \emph{\edist} if 
and only if there exists a minimal closed set $\B$
such that $E_1^\a = A(\mathcal{B})$ and $E_2^\a = 0$.
\end{definition}

\begin{theorem}~\label{thm:extreme}
The extreme-\dists are the extreme points of the \maxminpoly.
\end{theorem}

The proof uses similar techniques to our other proofs and is 
omitted.

\end{sloppy}

\end{document}